\documentclass[journal]{IEEEtran}
\usepackage{xcolor,soul,framed} 
\colorlet{shadecolor}{yellow}
\usepackage[pdftex]{graphicx}
\graphicspath{{../pdf/}{../jpeg/}}
\DeclareGraphicsExtensions{.pdf,.jpeg,.png}
\usepackage[cmex10]{amsmath}
\usepackage{algorithm}
\usepackage{algorithmic}
\usepackage{array}
\usepackage{tabularx}
\usepackage{mdwmath}
\usepackage{mdwtab}
\usepackage{eqparbox}
\usepackage{url}
\usepackage{amssymb}
\usepackage{amsthm}
\usepackage{amsmath}
\usepackage{amsthm}
\usepackage{booktabs}    
\usepackage{makecell}    
\usepackage{diagbox}
\usepackage{float}
\usepackage{titlesec}
\titleformat{\subparagraph}[runin]{\normalfont\normalsize\bfseries}{\thesubparagraph}{1em}{}
\titlespacing*{\subparagraph}{0pt}{3.25ex plus 1ex minus .2ex}{1em}
\definecolor{DBLACK}{RGB}{0,0,0}
\definecolor{DBLUE}{RGB}{0,0,255}
\newcommand{\td}[0]{\color{DBLACK}}

\theoremstyle{definition}

\theoremstyle{remark}

\theoremstyle{plain}
\newtheorem{lemma}{Lemma}
\newtheorem{proposition}{Proposition} 

\titlespacing{\section}{1pt}{1pt}{1pt}
\titlespacing{\subsection}{1pt}{1pt}{1pt}

\begin{document}
\bstctlcite{IEEEexample:BSTcontrol}
\title{\td{Revisiting Voltage and Synchronization Stability Analysis in Grid-Following Converter-Integrated Weak Grids: Insights from Non-Minimum-Phase Zeros } }

\author{
Fuyilong~Ma,~\IEEEmembership{Member,~IEEE},
Lidong~Zhang,~\IEEEmembership{Senior Member,~IEEE},
Wangqianyun~Tang,
Waisheng~Zheng,
Huanhai~Xin,~\IEEEmembership{Senior Member,~IEEE},
Linbin~Huang,~\IEEEmembership{Member,~IEEE},
{Lennart~Harnefors},~\IEEEmembership{Fellow,~IEEE}

\thanks{This work was supported by the Natural Science Foundation of China under Grant 52207110  (\textit {Corresponding author: Lidong Zhang}).

F. Ma, L. Zhang, W. Tang, and W. Zheng are with Electric Power Research Institute, China Southern Power Grid, Guangzhou 510080, China (e-mail: eemfyl@zju.edu.cn; \{zhangld4, tangwqy, zhengws\}@csg.cn).

H. Xin  and L. Huang are with  the College of Electrical Engineering, Zhejiang University, China. (e-mail: \{xinhh, hlinbin\}@zju.edu.cn). 

{L.Harnefors is with ABB Power Systems, 771 80 Ludvika, Sweden. (e-mail: lennart.harnefors@se.abb.com).}
}
}

\markboth{IEEE TRANSACTIONS ON POWER SYSTEMS}{}

\maketitle

\begin{abstract}
The increasing penetration of grid-following (GFL) converter-interfaced generators (CIGs) intensifies concerns over small-signal voltage and synchronization stability. While existing theories treat these two stability issues distinctly, practical wisdom in contrast employs a unified and static metric, short-circuit ratio (SCR), to assess both in weak grids.  
This paper aims to bridge this theory-practice gap by introducing the insight of non-minimum-phase (NMP) zeros. 
{\td First, we demonstrate that the two stability issues in weak grids
can be characterized within a common NMP-zero-based assessment framework: a zero at the origin corresponds to voltage instability, while low-frequency zeros impose fundamental constraints on synchronization dynamics. }
The traditional SCR is proven to be a special case of our proposed novel stability metric, NMP-zero (NMP-Z) factor, evaluated at the rated operating point. This establishes the theoretical foundation for the empirical success of SCR.
Building on this insight, we then develop a unified stability assessment method for multi-converter systems. 
The method retains the simplicity of SCR, requiring only the NMP-Z factor together with individual CIG dynamic models and enabling stability margin assessment under various operating points.
Our work provides a simple yet theoretically rigorous framework for stability analysis in CIG-integrated weak grids, with all theoretical findings and the proposed method validated through detailed time-domain simulations.
\end{abstract}


\section{Introduction}

\IEEEPARstart{T}{he} large-scale integration of renewable energy is fundamentally altering power system dynamics through the widespread displacement of synchronous generators by converter-interfaced generators (CIGs) \cite{gu2022power}. Grid-following (GFL) controllers are generally employed to these CIGs.
This structural shift has introduced two typical converter-driven instability phenomena: non-periodic voltage collapse and low-frequency (including sub-synchronous frequency) oscillations \cite{hatziargyriouDefinitionClassificationPower2021}.
The former is generally classified as small‑signal voltage stability (or short-term voltage stability) \cite{huangSmallDisturbanceVoltageStability2020}, while the latter is regarded as small‑signal synchronization stability, due to the central role of synchronization units such as  phase-locked loop (PLL) \cite{cheng2022real},\cite{wang2020grid}. These stability issues constrain the integration scale of CIGs, necessitating careful planning and operation of converter‑based power systems.
{\td To mitigate these weak-grid stability issues from the converter-control side, grid-forming (GFM) control and advanced GFL control strategies have recently received considerable attention. GFM converters regulate their terminal voltage and frequency and therefore exhibit voltage-source-like characteristics \cite{wang2020grid,zhang2010analysis}. For GFL converters, representative
improvements include PLL retuning and impedance reshaping, as well as PLL-less strategies \cite{mohammedComparisonPLLBased2022,
mohammedEnhancedFrequency2023}. These approaches can improve the performance of GFL converters in weak grids. Nevertheless, their stability still depends on the grid strength, operating point, frequency deviation, and controller structure. Therefore, system-level screening and quantitative assessment of these converter-driven stability issues remain essential.}

{\td To assess these converter-driven stability issues, existing literature provides distinct theoretical frameworks. On the one hand, voltage stability, often associated with loss of equilibrium near the maximum power transfer limit, is typically examined by static power‑flow methods such as \textit{PV} or \textit{QV} curves \cite{lund2022operating},\cite{peng2025measuring},\cite{machowski2020power}.
On the other hand, synchronization stability is inherently dynamic, governed by controllers like PLL, and thus requires dynamic methods such as state-space \cite{wangSmallSignalStabilityAnalysis2018} or impedance-based analysis \cite{harneforsModelingThreePhaseDynamic2007,morrisAnalysisControllerBandwidth2021,liuOscillatoryStabilityCriterion2018}.
State-space analysis provides modal information but requires detailed internal models, whereas impedance-based analysis supports black-box frequency-domain modeling and can predict converter-grid interaction stability using Bode plots or the generalized Nyquist criterion \cite{skogestad2005multivariable} or detailed time-domain simulation\cite{IEEEStandardInterconnection}. However, these dynamic methods typically require detailed models or broadband frequency-response information of the relevant system components. 
This inevitably increases computational burden compared to static methods. The characteristics of the above stability analysis methods are summarized in TABLE ~\ref{Table:stability_methods}}.

However, the theoretical division stands in stark contrast to practical wisdom, where system operators rely on a unified and static metric, short-circuit ratio (SCR), to fast screen stability issues\cite{conseilinternationaldesgrandsreseauxelectriquesConnectionWindFarms2016}. A low SCR signals a "\textit{weak grid}", where high network impedance relative to the CIG capacity intensifies risks of both small‑signal voltage and synchronization instability {\cite{IEEEStandardInterconnection}}. For example, a very weak grid with $\text{SCR}\approx1.0$ is well-recognized as a severe challenge to converter operation \cite{morrisAnalysisControllerBandwidth2021}. While this practical approach elegantly unifies both stability concerns, it still lacks a rigorous theoretical foundation.
\vspace{-1em}
\begin{table}[H]
\centering
{\td
\caption{Comparison of Stability-Analysis Methods}
}
\label{Table:stability_methods}

\scriptsize
\renewcommand{\arraystretch}{1.12}
\setlength{\tabcolsep}{2.2pt}

\resizebox{\columnwidth}{!}{%
\begin{tabular}{@{}p{1.45cm}p{3.05cm}p{3.05cm}p{2.25cm}@{}}
\toprule
Method
& Advantages
& Limitations
& Applications
\\
\midrule

Power-flow (PV,QV curves)
&
Simple; low computational burden
&
Neglects converter dynamics
&
Voltage-stability screening \cite{lund2022operating},\cite{peng2025measuring},\cite{machowski2020power}
\\

State-space (eigenvalues analysis)
&
Provides eigenvalues, damping, and participation
&
Requires detailed internal models
&
Small-signal modal analysis \cite{wangSmallSignalStabilityAnalysis2018} 
\\

Impedance-based
&
Supports black-box modeling and frequency-domain assessment
&
Requires broadband impedance/admittance models
&
Converter-grid interaction analysis \cite{harneforsModelingThreePhaseDynamic2007,morrisAnalysisControllerBandwidth2021,liuOscillatoryStabilityCriterion2018}
\\

Time-domain Simulation
&
Captures nonlinear and transient responses
&
Case dependent and computationally intensive
&
Detailed validation \cite{IEEEStandardInterconnection}
\\
\bottomrule
\end{tabular}%
}
\end{table}
\vspace{-1em}
Recognizing this gap, recent theoretical efforts revisit the wisdom of SCR. Several SCR variants are proposed for voltage stability analysis such as generalized SCR (gSCR) \cite{peng2025measuring,dong2018small} and site‑dependent SCR (SDSCR) \cite{wuAssessingImpactRenewable2018a}, whose assessment results are correlated with those of static methods for multi-converter power systems (MCPS). For small-signal synchronization stability, while SCR is recognized as a critical parameter \cite{morrisAnalysisControllerBandwidth2021}, its static nature may not fully capture the wide-band dynamics of CIGs. To address this, frequency‑dependent SCR metrics are constructed via impedance methods, such as virtual dynamic SCR (VDSCR) \cite{xiaoDesignOrientedSmallSignalStability2025} and impedance margin ratio (IMR) \cite{zhuImpedanceMarginRatio2024}. However, these advanced metrics require detailed frequency spectra from all system components (i.e., grid impedance and converters' impedance), thereby complicating practical implementation under varying operating conditions and undermining the original simplicity of SCR. 

A clear gap between theory and practice remains open. While the practical utility of SCR lies in its static simplicity for unified assessment, it lacks theoretical rigor; conversely, extended metrics improve dynamic fidelity but sacrifice practical simplicity. To resolve this, we aim to address two fundamental questions: 1) What common theoretical basis allows a metric like SCR to signal both voltage and synchronization instability? The answer will seek the common cause of different stability issues observed in weak grids; 2) How can a static metric inherently account for dynamic oscillatory constraints, enabling small-signal synchronization stability assessment? Pursuing this answer could pave the way to tackle the practical challenge posed by numerous CIGs with various operating points. 

To bridge this gap, we introduce the concept of non‑minimum‑phase (NMP) zeros from feedback control theory. NMP zeros (i.e., right-half-plane zeros) in the system transfer function are known to limit control bandwidth and degrade dynamic performance \cite{skogestad2005multivariable}.
{\td Their presence in single-converter systems (SCPS) and correlation with weak grids have been reported \cite{zhang2010analysis, wuImpactNonMinimumPhaseZeros2021b}. 
However, these studies do not establish a strict analytical relationship between the NMP-zero location and SCR. A NMP-zero-based illustration of stability issues remains obscure due to the high-order modeling, especially for multiple converters.}
Here, we reinterpret small-signal voltage and synchronization stability in weak grids through the insight of NMP zeros, and then develop a novel stability margin assessment method. The main contributions are summarized as follows.
\vspace{-0.1cm}
\begin{itemize}
    \item We formulate the SCR-based practice wisdom by introducing the concept of NMP zeros. Unlike previous theoretical frameworks that treat small‑signal voltage and synchronization stability separately, we propose a novel and unified stability metric termed the NMP-zero (NMP-Z) factor. It illustrates the location of NMP zeros to the corresponding stability limits. Notably, the traditional SCR and gSCR are proved to be special cases of the proposed NMP‑Z factor under rated operating conditions, thereby providing a rigorous control‑theoretic foundation for these practical indices in assessing combined stability risks.
\end{itemize}

\begin{itemize}
    \item We propose a static-metric-based method for small-signal voltage and synchronization stability assessment in MCPS. The proposed NMP-Z factor as a static metric retains the simplicity of SCR-based approaches, while being applicable under various non‑rated operating conditions. Moreover, unlike existing frequency-dependent SCR methods, our approach can be implemented via a constructed critical single-converter subsystem, requiring only the frequency spectrum of individual converters rather than those of all system components. This offers a more convenient framework for the analysis  of weak grids with numerous converters.
\end{itemize}
\vspace{-0.1cm}

The rest of this paper is organized as follows. Section II briefly reviews the impact of NMP zeros based on the feedback control theory. Section III takes SCPS to present the unified role of NMP zeros in small-signal voltage and synchronization stability,
and introduces the proposed NMP-Z factor to establish the theoretical basis for SCR. 
Section IV extends the proposed metric to MCPS and develops the unified stability margin assessment method. Section V provides simulation results to verify the effectiveness and capability of our proposed method.

\textit{Notation}: Let $\mathbb{R}$  and $\mathbb{C}$ denote the set of real numbers and complex numbers, respectively; $j = \sqrt { - 1} $ denotes the imaginary unit; superscript $^{*}$ denotes the complex conjugate; For a matrix $A$,  $A^T$ and $A^H$ denote its transpose and Hermitian transpose, respectively; $\det(A)$ denotes determinant of a matrix; $\lambda(A)$ and $\sigma(A)$denote its eigenvalues and singular values, respectively. We shall denote its minimum eigenvalue by $\underline{\lambda}(A)$ and the maximum singular value by $\overline{\sigma}(A)$. For a transfer function or matrix $G(s)$,
$\|{G}(s)\|_{\infty }=\max_\omega \overline{\sigma}\{{G}(j \omega)\}$.
When the context is clear, we may use $\mathbb{O}$ and $\mathbb{I}$ to present the zero and the identity matrix with proper dimensions.

\section{Brief Overview of Limitations Imposed by NMP Zeros}\label{Sec.2}
This section briefly introduces {\td 
the fundamental limitations imposed by NMP zeros on the stability margin of single-input-single-output (SISO) and multiple-input-multiple-output (MIMO) feedback systems.} 
{\td \subsection{Limitations in SISO Feedback Systems}}
Consider the SISO linear time-invariant feedback control system shown in Fig.~\ref{Fig.1-siso_diagram}, where the plant $G(s)$ and the controller $K(s)$ are driven by the reference $r$. The controlled output is $y$, the control input is $u$, and the tracking error $e = r - y$ serves as the input to the feedback controller. The loop transfer function is $L(s) = G(s)K(s)$. 
A central measure of closed-loop performance is the \textit{sensitivity function}
\begin{equation}\label{Eq:SISO_sensitivity}
    S(s) = (1+L(s))^{-1},
\end{equation}
 which is the transfer function from $r$ to $e$ and reflects how reference disturbances are attenuated \cite{havre1998effect}. Good closed-loop performance requires $|S(j\omega)| \ll 1$ over a wide frequency range of interest, indicating effective disturbance rejection and tracking  \cite{skogestad2005multivariable}. 
{\td The \textit{sensitivity peak} denoted by $\|S\|_\infty=\max_{\omega}|S(j\omega)|$ 
also serves as a key index: in the SISO case,
\begin{equation}\label{Eq:SISO_Nyquist_distance}
    \frac{1}{\|S\|_\infty}=\min_{\omega}|1+L(j\omega)|,
\end{equation}
which is the shortest distance from the Nyquist plot of $L(s)$ to the critical point $(-1,0)$. }

\begin{figure}[htbp]
\vspace{-0.3cm}
  \begin{center}
  \includegraphics[width=2.8in]{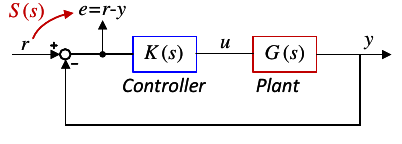}\\
\vspace{-0.3cm}
  \caption{Diagram of a feedback control system.}\label{Fig.1-siso_diagram}
  \end{center}
\end{figure}
\vspace{-0.5cm}

However, the presence of a \textit{NMP zero} in the plant imposes a fundamental lower bound on the achievable sensitivity peak, which must exceed unity under feedback, and may even lead to closed-loop instability when the resulting stability margin becomes insufficient \cite{havre1998effect}.
Suppose that the plant $G(s)$ has a real NMP zero at s=z, where $z>0$.
This gives the interpolation constraint $L(z) = 0
,S(z) = (1 + L(z))^{-1} = 1$. If the closed‑loop system is stable (i.e., $S(s)$ is analytic in the right‑half plane), 
the maximum-modulus theorem therefore yields the basic lower bound $\| S(s) \|_\infty \ge 1 $.
{\td A more restrictive limitation is obtained from the \textit{waterbed effect} \cite{havre1998effect}, as formalized in the following lemma.

\begin{lemma}[]\label{lem:siso_nmp_bound}
Consider the SISO feedback loop in Fig.~\ref{Fig.1-siso_diagram}. Suppose that $L(s)$ has no unstable poles and the closed loop is internally stable, and the plant has a real NMP zero $z>0$. If the sensitivity function satisfies the low-frequency specification
\begin{equation}\label{Eq:low_freq_spec_siso}
    |S(j\omega)|\leq M_L<1,\qquad |\omega|\leq \omega_L ,
\end{equation}
then the sensitivity peak satisfies
\begin{equation}\label{Eq:siso_nmp_lower_bound}
    \|S\|_\infty=\max_{\omega}|S(j\omega)|
    \geq \left(\frac{1}{M_L}\right)^{\frac{\theta_z}{\pi-\theta_z}}>1
\end{equation}
where $\theta_z=2\tan^{-1}\!\left(\frac{\omega_L}{z}\right)$.
\end{lemma}

\noindent\textit{Proof:}
This lemma is a consequence of Theorem~4 in \cite{freudenberg1985right} under the assumption that $L(s)$ has no unstable poles.\hfill$\blacksquare$

Lemma~\ref{lem:siso_nmp_bound} clearly shows the role of the NMP-zero location. Since $0<M_L<1$, the NMP zero $z$ closer to the origin leads to a larger $\theta_z$, and hence a larger lower bound on the sensitivity peak in \eqref{Eq:siso_nmp_lower_bound}. Equivalently, the achievable stability margin $1/\|S\|_\infty$ becomes increasingly restrictive as the NMP zero approaches the origin. Consequently, a near-origin NMP zero makes it extremely difficult to achieve low-frequency sensitivity reduction while maintaining a sufficient stability margin. }

{\td \subsection{Limitations in MIMO Feedback Systems}
For clarity, the above discussion has been presented for SISO feedback systems. The same fundamental limitation can be extended to MIMO feedback systems \cite{chen1995sensitivity,freudenberg1985right}.

For a square MIMO loop transfer matrix $L(s)$, define the sensitivity function $S(s)=(\mathbb{I}+L(s))^{-1}$. The sensitivity peak is also a stability margin index in MIMO case:
\begin{equation}\label{Eq:MIMO_singular_margin}
    \|S\|_\infty
    =\max_{\omega}\overline{\sigma}\!\left(S(j\omega)\right)
    =\frac{1}{\min_{\omega}\underline{\sigma}\!\left(I+L(j\omega)\right)} 
\end{equation}
where $\overline{\sigma}(\cdot)$ and $\underline{\sigma}(\cdot)$ denote the maximum and minimum singular values, respectively. Since the generalized Nyquist criterion \cite{skogestad2005multivariable} shows that , when $L(s)$ has no open-right-half-plane poles, closed-loop stability requires
\begin{equation}\label{Eq:MIMO_GNC}
    \det(I+L(j\omega))\neq 0, \forall \omega .
\end{equation}
$1/\|S\|_\infty$ measures the smallest distance, in the spectral norm, from $I+L(j\omega)$ to singularity. This is the MIMO counterpart of the SISO Nyquist distance in \eqref{Eq:SISO_Nyquist_distance}.

The sensitivity-peak limitation imposed by the NMP-zero location in SISO feedback systems remains valid for MIMO systems. The following lemma establishes this extension.

\begin{lemma}\label{lem:mimo_nmp_bound}
Suppose that $L(s)$ in the MIMO system has no unstable poles, the closed loop is internally stable, and $z>0$ is a real NMP zero with a unit left zero direction $w_z$, i.e.,
\begin{equation}\label{Eq:left_zero_direction}
    w_z^{H}L(z)=0,\qquad \|w_z\|_2=1 .
\end{equation}
If the low-frequency sensitivity reduction satisfies
\begin{equation}\label{Eq:low_freq_spec_mimo}
    \overline{\sigma}\!\left(S(j\omega)\right)\leq M_L<1,
    \qquad |\omega|\leq \omega_L ,
\end{equation}
then
\begin{equation}\label{Eq:mimo_nmp_lower_bound}
    \|S\|_\infty
    \geq
    \max_{\omega}
    \left|w_z^{H}S(j\omega)w_z\right|
    \geq
    \left(\frac{1}{M_L}\right)^{\frac{\theta_z}{\pi-\theta_z}}>1
\end{equation}
\end{lemma}
\noindent\textit{Proof:}
See Appendix~\ref{App:lemMIMO}.

Lemma~\ref{lem:mimo_nmp_bound} shows that the SISO sensitivity-peak limitation in Lemma~\ref{lem:siso_nmp_bound} remains applicable to MIMO feedback systems. Although the limitation is manifested along the zero direction, the NMP-zero location still determines its severity. Therefore, an NMP zero closer to the origin yields a larger lower bound on $\|S\|_\infty$ and hence a smaller achievable stability margin $1/\|S\|_\infty$. This forms the control-theoretic basis for the unified stability-margin analysis in Sections~III and~IV.
}

%
\section{Role of NMP zeros in Unified Stability Margin Assessment for SCPS}\label{Sec.3}
In this section, the insight from NMP zeros is employed to characterize the small-signal voltage and synchronization stability limits in a SCPS, based on the Jacobian transfer matrix model. Moreover, we propose the NMP-Z factor as a novel and unified stability metric. It helps to understand why the SCR-based practice wisdom is effective for screening these stability issues in weak grids. 

Without loss of generality, a three-phase grid-connected CIG is considered, as shown in Fig.~\ref{Fig.3-SCPS_diagram}. The CIG applies the most commonly-used control scheme designed in the controller’s dq frame (i.e., $d^cq^c$ frame, the angular frequency of which is the PLL’s angular frequency $\omega_{PLL}$). The scheme includes PLL, active power controller (APC), reactive power controller (RPC), dc voltage controller (DVC), and inner current controller \cite{camposNovelGramianbasedStructurepreserving2022}.
\subsection{System modeling Based on Jacobian Transfer Matrix}
To analyze the small-signal dynamics of the SCPS, the system is modeled using a Jacobian transfer matrix approach, which characterizes the power transfer dynamics in converter-based power systems \cite{zhang2010analysis}. As shown in Fig.~\ref{Fig.4-SCPS_Block}, the linearized SCPS is represented as a two-input, two-output feedback control system. The inputs are the variations in active and reactive power references , $[\Delta P_{\mathrm{ref}},\Delta Q_{\mathrm{ref}}]^{T}$, and the outputs are the corresponding power injections of CIG, $[\Delta P,\Delta Q]^{T}$. In this framework, the ac grid acts as the plant, while the CIG serves as the controller. The manipulated variables are the normalized voltage magnitude deviation and the phase angle deviation, $[U^{-1}\Delta U,\Delta \delta]^{T}$. The system dynamics are characterized by two Jacobian transfer matrices: $J_{\mathrm{NET}}(s)$ for the AC network and $J_{\mathrm{CIG}}(s)$ for the converter-interfaced generator. The detailed modeling procedure is described next. 
\vspace{-0.6cm}
\begin{figure}[H]
  \begin{center}
  \includegraphics[width=3.0in]{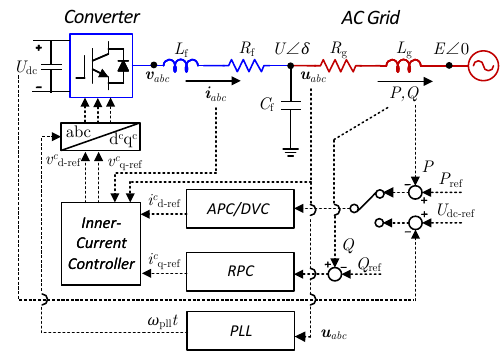}
  \caption{One-line diagram of a three-phase converter connected to the ac grid.}\label{Fig.3-SCPS_diagram}
  \end{center}
\end{figure}
\vspace{-0.9cm}
\begin{figure}[H]
  \begin{center}
  \includegraphics[width=3.0in]{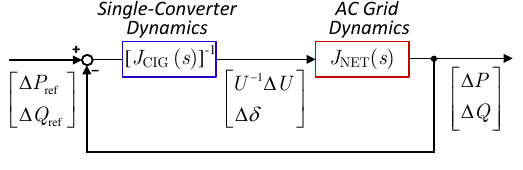}
  \caption{ Block diagram of the linearized SCPS model as a closed-loop feedback system.}\label{Fig.4-SCPS_Block}
  \end{center}
\end{figure}
\vspace{-0.7cm}

\textit{1) AC Grid Modeling in SCPS:}
Firstly, we start from the static case (i.e., $s=0$), and ${J_{{\rm{NET}}}}(0)$ is derived as the power-flow Jacobian matrix.  {\td It represents the linearization of the well-known active-reactive power transfer equations as follow
\begin{equation}\label{Eq.3}
\begin{aligned}
P &=
\frac{R\left(U^2-EU\cos\delta\right)+XEU\sin\delta}
{R^2+X^2},\\
Q &=
\frac{X\left(U^2-EU\cos\delta\right)-REU\sin\delta}
{R^2+X^2}.
\end{aligned}
\end{equation}

\begin{equation}\label{Eq.4}
J_{{\rm{NET}}}(0)
=
\left[
\begin{array}{cc}
P+U^2B_L\dfrac{\tau}{1+\tau^2}
&
-Q+U^2B_L\dfrac{1}{1+\tau^2}
\\[8pt]
Q+U^2B_L\dfrac{1}{1+\tau^2}
&
P-U^2B_L\dfrac{\tau}{1+\tau^2}
\end{array}
\right]
\end{equation}
where $R=R_g, X = \omega_0 L_g $ denote the line resistance and reactance with $\omega_0$ being the nominal grid angular frequency, respectively; $B_L = 1/X$ represents the line susceptance and $\tau=R/X$ represents the line resistance-to-reactance ratio.

To extend the analysis to the dynamic performance, the electromagnetic transients of the line should be considered, analytically yielding the Jacobian transfer matrix ${J_{{\rm{NET}}}}(s)$ as follow
\begin{equation}\label{Eq.5}
    \left[ \begin{array}{l}
\Delta P\\
\Delta Q
\end{array} \right] = {J_{\rm{NET}}}(s)\left[ \begin{array}{l}
{U^{ - 1}}\Delta U\\
\Delta \delta 
\end{array} \right]
\end{equation}
\begin{equation}\label{Eq.6}
    {J_{\rm{NET}}}(s) = \left[ {\begin{array}{*{20}{c}}
{P + {U^2}{B_L}\beta (s)}&{ - Q + {U^2}{B_L}\alpha (s)}\\
{Q + {U^2}{B_L}\alpha (s)}&{P - {U^2}{B_L}\beta (s)}
\end{array}} \right]
\end{equation}
where $\alpha(s)=
\frac{1}{(s/\omega_0+\tau)^2+1}$ and $\beta(s)=
\frac{s/\omega_0+\tau}
{(s/\omega_0+\tau)^2+1}$ . The derivation of ${J_{{\rm{NET}}}}(s)$ is detailed in Appendix~\ref{App:derivGrid}. 
Notably, by setting  $s=0$ in \eqref{Eq.6}, \eqref{Eq.4} can be recovered as the static form. }

\textit{2) Converter Modeling:}
As shown in Fig.~\ref{Fig.4-SCPS_Block}, the CIG is represented by a $2 \times 2$ Jacobian transfer matrix $J_{\mathrm{CIG}}(s)$ that incorporates the dynamics of the PLL and other controllers.
This formulation can be obtained from the well-established impedance/admittance model \cite{yang2019comparison} and shares its advantage: it also effectively overcomes the black-box modeling challenge posed by commercially confidential converters. This is because $J_{\mathrm{CIG}}(s)$ characterizes the broadband frequency-domain relationship between power disturbances and nodal voltage responses. Thus, its frequency-response $J_{\mathrm{CIG}}(j\omega)$ can be identified via frequency-scanning tests \cite{zhang2010analysis}.

\subsection{NMP Zeros in Grid with Single Converter}
Modeling the SCPS as a MIMO feedback system via Jacobian transfer matrices enables the analytical formulation of its NMP zeros. As outlined in Section II and depicted in Fig.~\ref{Fig.4-SCPS_Block}, these NMP zeros originate from the grid Jacobian matrix $J_{\mathrm{NET}}(s)$. The presence of NMP zeros imposes fundamental performance limitations, which are quantified by the maximum singular value of the sensitivity matrix, $\bar{\sigma}\{S(j\omega)\}$, and its peak $\| S(s) \|_\infty$. The sensitivity function of SCPS can be formulated as  $ S(s) = {({\mathbb{I}} + L(s))^{ - 1}}$, where $ L(s) = J_{\rm{{CIG}}}^{ - 1}(s){J_{\rm{NET}}}(s)$ represents the open-loop transfer fucntion matrix.

The high order of the MIMO transfer matrix $J_{\mathrm{NET}}(s)$ poses a significant challenge to deriving an analytical expression for its zeros. To overcome this, a grid \textit{complex Jacobian transfer matrix} $\tilde{J}_{\mathrm{NET}}(s)$ is introduced via a transformation of $J_{\mathrm{NET}}(s)$ into a complex form. This matrix compactly represents the relationship between complex power and voltage variations. Notably, the transformation preserves NMP zeros, ensuring $\tilde{J}_{\mathrm{NET}}(s)$ shares the same zeros as the original $J_{\mathrm{NET}}(s)$.

Specifically, (\ref{Eq.5}) can be first rewritten into the polar coordinate form as 
\begin{equation}\label{Eq.9}
\left[ \begin{array}{l}
{S^{ - 1}}\Delta S\\
\Delta \varphi 
\end{array} \right] = {T_\varphi }{J_{\rm{NET}}}(s)\left[ \begin{array}{l}
{U^{ - 1}}\Delta U\\
\Delta \delta 
\end{array} \right]
\end{equation}
where $S = \sqrt {{P^2} + {Q^2}} $  and $\varphi  = \arctan \frac{Q}{P}$  represent the apparent power and the power factor synchronization, respectively; ${S^{ - 1}}\Delta S$ represents the normalized variation of the apparent power relative to the given operating power;  ${T_\varphi }$ represents the polar transformation matrix which has
\[
\begin{bmatrix}
{S^{ -1}}\Delta S \\
\Delta \varphi 
\end{bmatrix} 
= 
S^{-1}
\begin{bmatrix}
\cos \varphi & \sin \varphi \\
-\sin \varphi & \cos \varphi
\end{bmatrix}
\begin{bmatrix}
\Delta P \\
\Delta Q
\end{bmatrix} 
:= {T_\varphi }
\begin{bmatrix}
\Delta P \\
\Delta Q
\end{bmatrix}.\]
Then, (\ref{Eq.9}) can be transformed into the complex-voltage and complex-power coordinate form, and the defined complex Jacobian transfer matrix of the grid has
\begin{equation}\label{Eq.10}
    \left[ \begin{array}{l}
{S^{ - 1}}\Delta S + j\Delta \varphi \\
{S^{ - 1}}\Delta S - j\Delta \varphi 
\end{array} \right] = {\tilde J_{\rm{NET}}}(s)\left[ \begin{array}{l}
{U^{ - 1}}\Delta U + j\Delta \delta \\
{U^{ - 1}}\Delta U - j\Delta \delta 
\end{array} \right]
\end{equation}
where ${\tilde J_{\text{NET}}}(s) = {T_j}[{T_\varphi }{J_{\text{NET}}}(s)]T_j^H$  represents the grid complex Jacobian transfer matrix by using the unitary transformation matrix $T_j=\frac{1}{\sqrt{2}} \begin{bmatrix}
 1 & j\\
 1 & -j
\end{bmatrix}$ , ${T_j}T_j^{H} = {\mathbb{I}}$ .

The above transformation yields ${\tilde J_{\rm{NET}}}(s)$  with a special property: its diagonal elements are all unity, and its off-diagonal entries form complex-conjugate pairs, as follow
\begin{equation}\label{Eq.11}
    {\tilde J_{\text{NET}}}(s) = \left[ {\begin{array}{*{20}{c}}
1&{\gamma (s){{\tilde S}^{ - 1}}{B_L}}\\
{{\gamma ^*}(s){{\tilde S}^{* - 1}}{B_L}}&1
\end{array}} \right]
\end{equation}
where $\tilde S: = (P + jQ){/}{U^2} = (S{{/}}{U^2}){e^{j\varphi }}$ represents the normalized complex power \cite{he2024complex}; $\gamma (s) = \beta (s) + j\alpha (s)$ and ${\gamma ^*}(s) = \beta (s) - j\alpha (s)$ are complex transfer functions \cite{harneforsModelingThreePhaseDynamic2007}.

Leveraging this property in \eqref{Eq.11}, the challenging process of finding the zeros $z$ of MIMO transfer matrices by solving its determinant equation can be simplified as
\begin{equation}\label{Eq.12}
\hspace*{-0.3em}
    0\! = \!\det (J_{\text{NET}}(z)) = \det (\tilde J_{\rm{NET}}(z))    
       = 1 - \gamma (z){\gamma ^*}(z){\rho_z}
\end{equation}
\vspace{-1em} 
\begin{equation}\label{Eq.13}
       {\rho _z} = ({\tilde S^{-1}}{B_L}) \cdot ({\tilde S^{*-1}}{B_L}) > 0
\end{equation}
where the second equality in (\ref{Eq.12}) holds thanks to the non-singular transformation matrices in  (\ref{Eq.10}), while the third equality is derived using the Schur complement \cite{zhang2006schur};
{\td $\gamma(z)\gamma^*(z)={\beta ^2}(z)+{\alpha ^2}(z)
=\frac{\omega_0^2}{(z+\tau\omega_0)^2+\omega_0^2},$ ;
${\rho _z}$  represents an important positive parameter associated with the solution for NMP zeros, and will be discussed in Section III.C. 

Solving \eqref{Eq.12} and \eqref{Eq.13} gives that $\tilde{J}_{\mathrm{NET}}(s)$ and $J_{\text{NET}}(s)$ share a pair of real zeros. when $\rho_z>1+\tau^2$, yielding an NMP zero locates in right-half plane
\begin{equation}\label{Eq.14}
   z = -\tau\omega_0 \pm \omega_0\sqrt{\rho_z-1},
   z_{\text{NMP}} =-\tau\omega_0+\omega_0\sqrt{\rho_z-1}.
\end{equation}
Particularly, for the commonly adopted approximation $R\ll X,\tau \approx 0$ \cite{machowski2020power}, the NMP zero can approximately reduce to $z_{\text{NMP}} = \omega_0\sqrt{\rho_z-1}$.}

\subsection{Role of NMP Zero and SCR in Unified Stability Margin}
We now illustrate the role of NMP zeros in the small-signal stability analysis of SCPS, encompassing both voltage and synchronization stability and provide a theoretical explanation for the SCR-based experience.

\noindent\textbf{Definition 1.} (Non-minimum-phase-zero factor)
\textit{Given the NMP zero of the grid Jacobian transfer matrix, formulated as ${z_{\rm{NMP}}} = -\tau\omega_0+\omega_0\sqrt{\rho_z-1} $ , the parameter $\rho_z$  is defined as the grid NMP-zero factor (abbreviated to \textit{NMP-Z factor}). This key parameter pinpoints the location of the NMP zero and is crucial for the small-signal stability and dynamic performance of converter-based power systems.}

Based on Definition 1, the NMP-Z factor for SCPS is given in \eqref{Eq.13}. It serves a unified stability metric for small-signal voltage and synchronization stability, as illustrated below.

\textit{1) Small-Signal Voltage Stability:}
Its boundary can be characterized by 
{\td ${\rho_z}=1+\tau^2$, or equivalently  $z_{\text{NMP}}=0$ (the NMP zero at the origin). 
This condition signifies the singularity of the static power-flow Jacobian matrix ($\det J_{\text{NET}}(0)=0$ in \eqref{Eq.12}), which is a direct indicator of voltage instability \cite{lund2022operating,machowski2020power}. The following criterion should be satisfied to ensure the stable quasi-static or static operation of SCPS: 
\begin{equation}\label{Eq.15}
    {z_{\text{NMP}}} > 0,{\rho _z} > 1+\tau^2
\end{equation}
Particularly, the boundary approximately reduces to ${\rho_z}=1$ when $R\ll X,\tau \approx 0$ hold.
}

\textit{2) Small-Signal Synchronization Stability:} It concerns oscillations primarily induced by the weak-grid limitations for PLL dynamics \cite{cheng2022real}. The NMP-Z factor  ${\rho _z}$ in \eqref{Eq.13} integrates the key weak-grid factors such as operating power and line reactance. As the corresponding $z_{\mathrm{NMP}}$ approaches the origin, the limitation within the PLL bandwidth (encompassing low and sideband frequencies \cite{wang2018harmonic}) intensifies, resulting in a high sensitivity peak and eventual oscillation  instability. Thus, the following criterion should be satisfied: 
\begin{equation}\label{Eq.16}   
    {z_{\text{NMP}}} > {z_c},{\rho _z} > \rho _{zc}
\end{equation}
where  ${z_c}$  and $ {\rho _{zc}}$  represent the thresholds of NMP zero and NMP-Z factor on the small-signal synchronization stability boundary (i.e., $1/||S(s)|{|_\infty } = 0$), linked to control parameters of the given CIG.

The stability criteria derived from the proposed NMP-Z factor establish a theoretical basis for the popular SCR-based practical experience. The SCR for SCPS in Fig.~\ref{Fig.3-SCPS_diagram} can be expressed 
\begin{equation}\label{Eq.17}
    SCR = \frac{{{S_{ac}}}}{{{P_N}}} = \frac{{U_N^2/{Z_{{\rm{th}}}}}}{{{P_N}}} = P_N^{ - 1}{B_L}
\end{equation}
where $S_{ac}$ is the short-circuit capacity from the grid at the CIG bus; $U_N = 1.0 \text{p.u.}$ represents the nominal voltage; $Z_{\rm{th}}\approx X=1/B_L$ represents the grid Thevenin impedance magnitude;  $P_N$ represents the rated active power of CIG.  

\noindent\textbf{Remark 1.} (Why SCR works)
\textit{From the perspective of NMP zeros, SCR emerges a special case of the proposed NMP-Z factor at the rated operating point of CIG ($P = P_N$, $Q = 0$, $U = U_N = 1.0$ p.u.). This fundamental connection explains the broad empirical utility of SCR in screening both small-signal voltage and synchronization stability limits. Substituting these conditions into \eqref{Eq.13} and using \eqref{Eq.17} yields}
\begin{equation}\label{Eq.18}
    {\rho _z} = (P_N^{ - 1}{B_L}) \cdot (P_N^{ - 1}{B_L}) = SCR^2
\end{equation}
\textit{Therefore, the NMP-Z factor $\rho_z$ also generalizes SCR, with validity extending to both rated and off-nominal operating conditions.}

{\td It should be mentioned that the proposed NMP-Z-factor-based framework is developed primarily for GFL converters operating in weak grids. 
In contrast, GFM converters may exhibit voltage source characteristics and introduce the different instability mechanisms in strong grids, where the NMP zero is generally located farther away from the origin and therefore impose a much less restrictive limitation on the achievable control performance.}

\noindent\textbf{Example 1.}(A Single-Converter Power System) 
\textit{The SCPS in Fig. 3 is considered with converter parameters given in Appendix~\ref{App:paraSystem}. A initial non‑rated operating point is set as $P=0.8$ p.u., $Q=-0.2$ p.u., $U=0.95$ p.u., and lossless line reactance $X=0.5$ p.u. for the simplicity (equivalent to $\text{SCR}=2.0$ at the rated condition). By varying the operating point, both small‑signal synchronization and voltage instabilities are reported, as shown in Fig.~\ref{Fig.5-SCPS_PQplot} and Fig.~\ref{Fig.6-SCPS_Senplot}, respectively.}

\textit{In Fig.~\ref{Fig.5-SCPS_PQplot}(a), the active power is stepped from $P=0.8$ p.u. to $1.05$ p.u. at 1s intervals (PLL bandwidth = 15 Hz). At $P=1.05$ p.u., a 12.8 Hz power oscillation emerges, consistent with reported observations of synchronization instability \cite{cheng2022real}. 
In Fig.~\ref{Fig.5-SCPS_PQplot}(b), the reactive power is stepped down from $Q=-0.2$ p.u. to $-0.38$ p.u. The low reactive power at $Q=-0.38$ p.u. triggers a non‑periodic instability, reflecting the typical voltage instability pattern \cite{lund2022operating}.}

\textit{These time‑domain responses are explained by the sensitivity peak induced by NMP zeros. Fig.~\ref{Fig.6-SCPS_Senplot}(a) shows the Bode magnitude of $\bar\sigma\{S(j\omega)\}$ for $P=0.8$ p.u. and $P=1.05$ p.u. As $\rho_z$ decreases from 4.79 to 1.26, the NMP zero moves closer to the origin, raising $\bar\sigma\{S(j\omega)\}$ and pushing $\|S(s)\|_\infty$ near 20 dB ($1/||S(s)|{|_\infty } \approx 0$). The resulting high sensitivity peak near 12.8 Hz correlates with the observed oscillation in  Fig.~\ref{Fig.5-SCPS_PQplot}(a), demonstrating how the NMP zero located near the fundamental angular frequency sideband  (i.e., $z_\text{NMP}=0.51\omega_0$) constrains PLL dynamics. 
Fig.~\ref{Fig.6-SCPS_Senplot}(b) compares $\bar\sigma\{S(j\omega)\}$ for $Q=-0.2$ p.u. and $Q=-0.38$ p.u. With $\rho_z = 1.01 \approx 1$ (for $Q=-0.38$ p.u.), the NMP zero lies very near the origin (i.e., $z_\text{NMP}=0.1\omega_0$), causing $\bar\sigma\{S(j\omega)\}$ to exceed 0 dB near the zero frequency. This indicates the  inability within the quasi-static or static time scale, matching the power-step response in Fig.~\ref{Fig.5-SCPS_PQplot}(b).}

\textit{{\td In addition, throughout these operating-point changes, SCR emerges a special case of the proposed NMP-Z factor at the rated operating point and thus remains constant at 2.0. It fails to indicate the stability limits, whereas the proposed NMP-Z factor can capture them under non‑rated conditions.}}
\vspace{-1.3em}
\begin{figure}[H]
\centering
\includegraphics[width=3.0in]{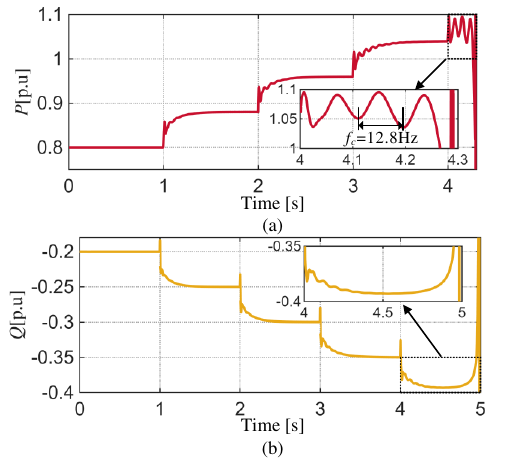}
\caption{ Time-domain responses of SCPS under step changes in active and reactive power. (a) active power output with step increases from $P=0.8$p.u. to $P=1.05$p.u. and small-signal synchronization instability induces oscillations. (b) reactive power output with step decreases  from $Q=-0.2$p.u. to $Q=-0.38$p.u. and small-signal voltage instability emerges.}\label{Fig.5-SCPS_PQplot}
\end{figure}
\vspace{-2.3em}
\begin{figure}[H]
\centering
\includegraphics[width=3.0in]{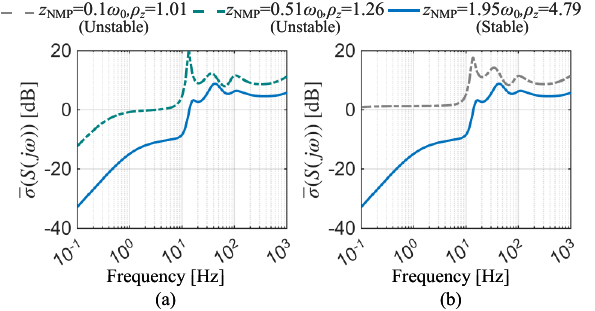}
\caption{Bode magnitude plots  of  $\bar \sigma \{ S(j\omega )\} $ for different operating points. (a)  comparison for  $P=0.8$p.u. and $P=1.05$p.u.; (b)  comparison for  $Q=-0.2$p.u. and $Q=-0.38$p.u., respectively.}\label{Fig.6-SCPS_Senplot}
\end{figure}
\vspace{-0.5cm}
\section{Unified Stability Margin Assessment Extended for MCPS with NMP Zeros}
This section extends the results from SCPS to MCPS and develops a unified stability margin assessment method based on the proposed NMP-Z factor. 
With loss of generality, consider the MCPS with $n$ converter buses coupled though the ac network, as shown in Fig.~\ref{Fig.7-MCPS_Diagram}. The network can be abstracted as an undirected graph ${\cal G} = G({\cal V},{\cal E})$, where $\cal V$ and $\cal E$ denote the sets of buses and transmission lines, respectively \cite{dorflerKronReductionGraphs2013b}. Each bus $i \in {\cal V}$ can be characterized by by $({P_i},{Q_i},{U_i},{\delta _i})$, representing its active power, reactive power, voltage magnitude and phase angle.
The network topology and line reactance reciprocals are encoded in the 
grounded Laplacian matrix $\boldsymbol{B} \in \mathbb{R}^{n \times n}$ 
\cite{dorflerKronReductionGraphs2013b}. Let ${\cal E}$ denote the set of 
lines. For each line $(i,j)\in{\cal E}$, the off-diagonal entries of 
$\boldsymbol{B}$ are defined as
$B_{ij}=B_{ji}=-1/X_{ij}<0$, where $X_{ij}$ is the reactance of line 
$(i,j)$. For $(i,j)\notin{\cal E}$, one has $B_{ij}=B_{ji}=0$. The diagonal 
entries are given by
$B_{ii}=\sum_{j\neq i}(-B_{ij})+b_{i0}>0$, where $b_{i0}$ denotes the shunt 
susceptance at node $i$.
{\td For analytical tractability in the MCPS, all transmission lines are assumed to have an identical resistance-to-reactance ratio, i.e.,$\tau=R_{ij}/X_{ij},\forall (i,j)\in{\cal E}$. }
\vspace{-0.47em}
\begin{figure}[H]
\centering
\hspace*{-1em}  
\includegraphics[width=3.5in]{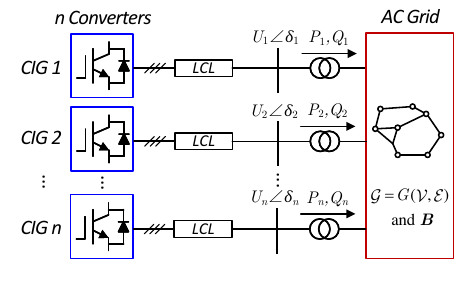}
\vspace{-0.6cm}
\caption{One-line diagram of multiple converters connected to the ac grid.}\label{Fig.7-MCPS_Diagram}
\end{figure}
\vspace{-0.5cm}
\begin{figure}[H]
\centering
\includegraphics[width=3.5in]{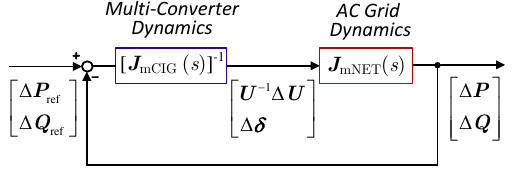}
\vspace{-0.5cm}
\caption{Block diagram of the linearized MCPS model as a closed-loop feedback system. }\label{Fig.8-MCPS_Block}
\end{figure}
\subsection{Jacobian Transfer Matrix of MCPS}
Following the approach in Section III, the small‑signal dynamics of the MCPS are modeled using Jacobian transfer matrices within a MIMO feedback loop, as shown in Fig.~\ref{Fig.8-MCPS_Block}. The input vectors are the active- and reactive-power reference disturbances $\Delta \boldsymbol{P}_{\text{ref}}$ and $\Delta \boldsymbol{Q}_{\text{ref}}$; the output vectors are the corresponding converter power variations $\Delta \boldsymbol{P} = [\Delta P_1,\dots,\Delta P_n]^{T}$ and $\Delta \boldsymbol{Q} = [\Delta Q_1,\dots,\Delta Q_n]^{T}$. The manipulated variables are the normalized voltage magnitude vector $\boldsymbol{U}^{-1}\Delta\boldsymbol{U} = [U_1^{-1}\Delta U_1,\dots,U_n^{-1}\Delta U_n]^{T}$ and the angle vector $\Delta\boldsymbol{\delta} = [\Delta\delta_1,\dots,\Delta\delta_n]^{T}$. The dynamics are characterized by two $2n \times 2n$ Jacobian transfer matrices: $\boldsymbol{J}_{\mathrm{mNET}}(s)$ for the AC network and $\boldsymbol{J}_{\mathrm{mCIG}}(s)$ for CIGs. The detailed derivation follows.

The grid Jacobian transfer matrix in MCPS can be derived
\begin{equation}\label{Eq.19}
\left[ \begin{array}{l}
\Delta \boldsymbol{P}\\
\Delta \boldsymbol{Q}
\end{array} \right] = {\boldsymbol{J}_{\text{mNET}}}(s)\left[ \begin{array}{l}
{\boldsymbol{U}^{ - 1}}\Delta \boldsymbol{U}\\
\Delta \boldsymbol{\delta} 
\end{array} \right]
\end{equation}

\begin{equation}\label{Eq.20}
\hspace*{-0.5em}  
\begin{array}{c}
\boldsymbol{J}_{\text{mNET}}(s) \!=\! \left[ {\begin{array}{*{20}{c}}
{{\mathop{\rm Re}\nolimits} \{ \boldsymbol{\tilde{Y}}\} }&{ - {\mathop{\rm Im}\nolimits} \{ \boldsymbol{\tilde{Y}}\} }\\
{{\mathop{\rm Im}\nolimits} \{ \boldsymbol{\tilde{Y}}\} }&{{\mathop{\rm Re}\nolimits} \{ \boldsymbol{\tilde{Y}}\} }
\end{array}} \right]\left[ {\begin{array}{*{20}{c}}
{\beta (s)\mathbb{I}}&{\alpha (s)\mathbb{I}}\\
{\alpha (s)\mathbb{I}}&{ - \beta (s)\mathbb{I}}
\end{array}} \right]\\
 + \left[ {\begin{array}{*{20}{c}}
\boldsymbol{P}&{ - \boldsymbol{Q}}\\
\boldsymbol{Q}&\boldsymbol{P}
\end{array}} \right]
\end{array}
\end{equation}
where $\tilde{\boldsymbol{Y}} \in \mathbb{C}^{n \times n}$ is the grid complex susceptance matrix, with elements $\tilde{Y}_{ij} = U_i U_j B_{ij} e^{j\delta_{ij}}$, $\delta_{ij} = \delta_i - \delta_j$. Its real and imaginary parts are $\operatorname{Re}\{\tilde{Y}_{ij}\} = U_i U_j B_{ij} \cos\delta_{ij}$ and $\operatorname{Im}\{\tilde{Y}_{ij}\} = U_i U_j B_{ij} \sin\delta_{ij}$, respectively. The converter active- and reactive-power outputs form the diagonal matrices $\boldsymbol{P} = \operatorname{diag}\{P_i\}$ and $\boldsymbol{Q} = \operatorname{diag}\{Q_i\}$. The detailed derivation of \eqref{Eq.19} and \eqref{Eq.20} is provided in Appendix ~\ref{App:derivGrid}.

The Jacobian transfer matrix of CIGs can be expressed as a blocked matrix 
\begin{equation}\label{Eq.21}
\hspace*{-0.5em}  
\boldsymbol {J}_{\text{mCIG}}(s) = \left[ {\begin{array}{*{20}{c}}
{\text{diag}\{ {{J}_{\text{CIG}i,11}}(s)\} }&{\text{diag}\{ { {J}_{\text{CIG}i,12}}(s)\} }\\
{\text{diag}\{ { {J}_{\text{CIG}i,21}}(s)\} }&{\text{diag}\{ { {J}_{\text{CIG}i,22}}(s)\} }
\end{array}} \right]
\end{equation}
where $J_{\text{CIG}i,11}(s)$, $J_{\text{CIG}i,12}(s)$, $J_{\text{CIG}i,21}(s)$, and $J_{\text{CIG}i,22}(s)$ denote the entries of the $2 \times 2$ Jacobian transfer matrix $J_{\text{CIG}i}(s)$ of the $i$-th converter.

\subsection{NMP Zeros in Grid with Multiple Converters}
As the plant in the MCPS feedback loop shown in Fig.~\ref{Fig.8-MCPS_Block}, the NMP zeros of ${\boldsymbol{J}_\text{mNET}}(s)$  impose the fundamental limitations on the system dynamic performance and stability. To obtain their analytical expressions, we extend the approach in  Section III to MCPS and introduce the corresponding complex Jacobian transfer matrix form $ \boldsymbol{\tilde{J}}_\text{mNET}(s)$:
\begin{equation}\label{Eq.22}
\hspace*{-0.4em}  
\left[ \begin{array}{l}
{\boldsymbol{S}^{ - 1}}\Delta\boldsymbol{S} + j\Delta \boldsymbol{\varphi} \\
{\boldsymbol{S}^{ - 1}}\Delta\boldsymbol{S} - j\Delta \boldsymbol{\varphi}  
\end{array} \right] \!= \!{\boldsymbol{\tilde{J}}_\text{mNET}}(s)\left[ \begin{array}{l}
{\boldsymbol{U}^{ - 1}}\Delta \boldsymbol{U} + j\Delta \boldsymbol{\delta} \\
{\boldsymbol{U}^{ - 1}}\Delta \boldsymbol{U} - j\Delta \boldsymbol{\delta} 
\end{array} \right]
\end{equation}
where
\begin{equation}\label{Eq.23}
\begin{array}{c}
\boldsymbol{\tilde {J}}_\text{mNET}(s) = {T_{mj}}({T_{m\varphi }}{\boldsymbol{J}_\text{mNET}}(s))T_{mj}^H\\
 = \left[ {\begin{array}{*{20}{c}}
{ \mathbb{I}}&{\gamma (s){{\boldsymbol{\tilde{S}} }^{ - 1}}\boldsymbol{\tilde{Y}}}\\
{{\gamma ^*}(s){{\boldsymbol{\tilde{S}}}^{* - 1}}{{\boldsymbol{\tilde {Y}}}^*}}&{\mathbb{I}}
\end{array}} \right]
\end{array}
\end{equation}
represents the grid complex Jacobian transfer matrix transformed  from  $ \boldsymbol{J}_\text{mNET}(s)$; the transformation matrices are ${T_{m\varphi }} = \left[ {\begin{array}{*{20}{c}}
{{\boldsymbol{S}^{ - 1}}}&\mathbb{O}\\
\mathbb{O}&{{\boldsymbol{S}^{ - 1}}}
\end{array}} \right]\left[ {\begin{array}{*{20}{c}}
{{\rm{diag}}\{ \cos {\varphi _i}\} }&{{\rm{diag}}\{ \sin {\varphi _i}\} }\\
{ - {\rm{diag}}\{ \sin {\varphi _i}\} }&{{\rm{diag}}\{ \cos {\varphi _i}\} }
\end{array}} \right]$ and the unitary matrix ${T_{{\rm{m}}j}} = \frac{1}{{\sqrt 2 }}\left[ {\begin{array}{*{20}{c}}
{{\mathbb{I}}}&{j{\mathbb{I}}}\\
{{\mathbb{I}}}&{ - j{\mathbb{I}}}
\end{array}} \right]$(i.e., ${T_{{\rm{m}}j}}T_{{\rm{m}}j}^H = \mathbb{I}$);  The diagonal matrix of  complex power is $\tilde{\boldsymbol{S}} = \boldsymbol{P} + j\boldsymbol{Q}$, with $\boldsymbol{S} = \operatorname{diag}\{S_i\}$ being the apparent power matrix; $\Delta \boldsymbol{S} = {[\Delta {S_1},...,\Delta {S_n}]^T}$ and $\Delta \boldsymbol{\varphi}  = {[\Delta {\varphi _1},...,\Delta {\varphi _n}]^T}$ represent the converters’ variation vectors of the apparent power and the power factor, respectively.

\begin{lemma}[]\label{lem:eigenvalue_Heq}
The matrix $\boldsymbol{H}_{\mathrm{eq}} := \boldsymbol{\tilde{S}}^{-1}\boldsymbol{\tilde{Y}}\boldsymbol{\tilde{S}}^{*-1}\boldsymbol{\tilde{Y}}^{*}$ is diagonalizable with positive, simple eigenvalues. Hence, there exists a nonsingular matrix $\boldsymbol{W}$ such that
\begin{equation}\label{Eq.24}
    {\boldsymbol{W}^{ - 1}}{\boldsymbol{H}}_\text{eq}\boldsymbol{W} = {\boldsymbol{\Lambda}_\text{eq} } = {\text{diag}}\{ {\lambda _i}\} .
\end{equation}
where the eigenvalues satisfy $0 < \lambda_1 \le \dots \le \lambda_n$.
\end{lemma}
\noindent\textit{Proof:}
See Appendix~\ref{App:lemHeq}.

Combing Lemma ~\ref{lem:eigenvalue_Heq}, the challenge of finding NMP zeros $z$ by solving the determinant of  $ \boldsymbol{\tilde{J}}_\text{mNET}(s)$  in (\ref{Eq.23}), originally involving a $2n\times2n$ high-order matrix can be decoupled into $n$ scalar equations:
\begin{equation}\label{Eq.25}
    \begin{aligned}
0 &= \det ( \boldsymbol{J}_\text{mNET}(z))  = \det (\boldsymbol{\tilde{J}}_\text{mNET}(z)) \\
  &= \det ( {\mathbb{I}} - \gamma (z){\gamma ^*}(z){\boldsymbol{H}_\text{eq}})\\
  &= \det\boldsymbol{W} ( {\mathbb{I}} - \gamma (z){\gamma ^*}(z){\boldsymbol{\Lambda}_\text{eq}})\boldsymbol{W}^{-1}\\
  &= \prod\nolimits_{i = 1}^n {(1 - \gamma (z){\gamma ^*}(z){\lambda _i})} 
    \end{aligned}
\end{equation}
where the second equality follows from the non-singular transformation matrices in \eqref{Eq.23}, the third from the Schur complement, and the fourth and fifth from Lemma ~\ref{lem:eigenvalue_Heq}.

{\td Then, $\boldsymbol{\tilde{J}}_\text{mNET}(s)$ (or ${\boldsymbol{J}}_\text{mNET}(s)$ ) has $n$ pairs of real zeros and yields  $n$ NMP zeros when $\lambda_1>1+\tau^2$:
\begin{equation}\label{Eq.26}
    {z_i} =  -\tau\omega_0\pm {\omega _0}\sqrt {{\lambda _i} - 1} ,{z_{\text{NMP}i}} = -\tau\omega_0+{\omega _0}\sqrt {{\lambda _i} - 1} 
\end{equation}
where $z_{\text{NMP}i}$ denotes the $i$-th NMP zero of MCPS. }

\subsection{Stability Margin based on Critical NMP Zero}
As established in Section~\ref{Sec.2}, the NMP zeros of the grid Jacobian transfer matrix approaching the origin fundamentally constrain the dynamic performance of MCPS, and consequently reduce the stability margin. We now propose a critical single-converter subsystem for MCPS to identify the critical NMP zero and its threshold, which determine small-signal voltage and synchronization stability boundaries.

\noindent\textbf{Definition 2.}(Critical NMP zero and NMP-Z factor for MCPS) \textit{Generalizing Definition 1 to the MCPS, the NMP‑Z factor is defined as}
\begin{equation}\label{Eq.27}
    \rho_z:=\lambda_1=\underline{\lambda}({\boldsymbol{H}_\text{eq}})= 
    \underline{\lambda}( {\boldsymbol{\tilde {S}}^{ - 1}}\boldsymbol{\tilde{Y}}{\boldsymbol{\tilde {S}}^{* - 1}}{\boldsymbol{\tilde {Y}}^*})
\end{equation}
\textit{where $\rho_z$ locates the critical NMP zero denoted by $z_{\mathrm{NMP}1}$, the one closest to the origin among all $n$ NMP zeros, and thus imposes the most restrictive sensitivity-peak limitation. }
\begin{figure}[htbp]
\centering
\includegraphics[width=3.5in]{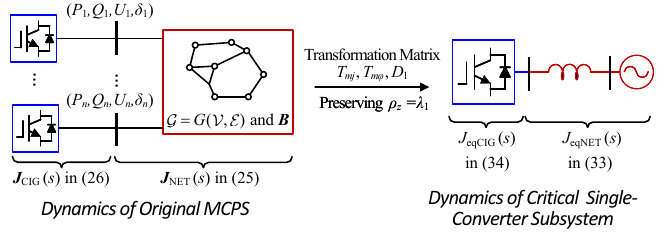}
\vspace{-2.3em}
\caption{Diagram of critical single-converter subsystem for MCPS. }\label{Fig.9-Subsystem_MCPS}
\end{figure}

To isolate this critical NMP zero, define 
$D_1 = [w_1,\mathbb{O};\mathbb{O},w_1^*],D_1^H{D_1} = \mathbb{I}$, where $w_1\in \mathbb{C}^{n},w_1^Hw_1=1$ is the normalized eigenvector of  $\boldsymbol{H}_\text{eq}$ corresponding to $\rho_z=\lambda_1$.  
As shown in Fig.~\ref{Fig.9-Subsystem_MCPS},
the equivalent grid and converter Jacobian transfer matrices of the resulting critical single-converter subsystem are 
\begin{equation}\label{Eq.28}
\bar J_{\text{eqNET}}(s) = D_1^H \boldsymbol{\tilde {J}}_\text{mNET}(s) D_1 = \left[ {\begin{array}{*{20}{c}}
1&{\gamma (s){\mu _1}}\\
{{\gamma ^*}(s)\mu _1^*}&1
\end{array}} \right],
\end{equation}
\begin{equation}\label{Eq.29}
\begin{aligned}
\bar J_{\text{eqCIG}}(s)
&=D_1^H\boldsymbol{\tilde J}_{\text{mCIG}}(s)D_1 \\
&=\left[
\begin{array}{@{}c@{\;}c@{}}
\sum\nolimits_i |w_{1i}|^2
\tilde J_{\text{CIG}i,11}(s)
&
\sum\nolimits_i w_{1i}^{*2}
\tilde J_{\text{CIG}i,12}(s)
\\[2pt]
\sum\nolimits_i w_{1i}^{2}
\tilde J_{\text{CIG}i,21}(s)
&
\sum\nolimits_i |w_{1i}|^2
\tilde J_{\text{CIG}i,22}(s)
\end{array}
\right],
\end{aligned}
\end{equation}
where 
$w_{1i}$  represents the $i$-th entry of  $w_1$;
${\mu_1} = w_1^H{\boldsymbol{\tilde S}^{-1}}{\boldsymbol{\tilde Y}}w_1^*$; 
$\boldsymbol{\tilde J}_{\text{mCIG}}(s) = {T_{mj}}({T_{m\varphi }}{\boldsymbol{J}_\text{mCIG}}(s))T_{mj}^H$ represents the complex Jacobian transfer matrix of the converters transformed from  
$\boldsymbol{J}_\text{mCIG}(s)$,
and ${\tilde J}_{\text{CIG}i,nm}(s)$  
denotes the $(n,m)$ entry of the
$i$th converter complex Jacobian transfer matrix. 
Thus, the equivalent converter in the critical subsystem is a modal combination of all converters weighted by $w_1$, whereas the equivalent grid retains the same $2\times2$ structure as the SCPS grid model in \eqref{Eq.11}.
{\td 
\begin{proposition}[Preservation of the NMP-Z factor]
\label{prop:preserveNMPzero}
The critical single-converter subsystem preserves the NMP-Z factor of the original MCPS associated with $\lambda_1$ as follow
\begin{equation}\label{Eq.prop1}
   \rho_z=\lambda_1=\mu_1\mu_1^*.
\end{equation}
\end{proposition}
\noindent\textit{Proof:}
See Appendix~\ref{App:propNMP}.

We next relate the critical subsystem to the closed-loop stability
boundary of MCPS for determining the threshold of the critical NMP-Z factor.
Define the characteristic matrices of the critical subsystem and the
original MCPS as
\begin{equation}\label{Eq.characteristicMatrices}
\begin{aligned}
\bar C_{\mathrm{eq}}(s)
&:=\bar J_{\mathrm{eqCIG}}(s)
  +\bar J_{\mathrm{eqNET}}(s),\\
\boldsymbol{\tilde C}_{\mathrm{m}}(s)
&:=\boldsymbol{\tilde J}_{\mathrm{mCIG}}(s)
  +\boldsymbol{\tilde J}_{\mathrm{mNET}}(s).
\end{aligned}
\end{equation}
and the corresponding sensitivity functions
\begin{equation}\label{Eq.sensitivityMCPS}
\begin{aligned}
S_{\mathrm{eq}}(s)
&:=
(
\mathbb I+
\bar J_{\mathrm{eqNET}}(s)
\bar J_{\mathrm{eqCIG}}^{-1}(s)
)^{-1},\\
\boldsymbol S_{\mathrm{m}}(s)
&:=
(
\mathbb I+
\boldsymbol{\tilde J}_{\mathrm{mNET}}(s)
\boldsymbol{\tilde J}_{\mathrm{mCIG}}^{-1}(s)
)^{-1}.
\end{aligned}
\end{equation}

Provided that
$\boldsymbol{\tilde J}_{\mathrm{mCIG}}(j\omega)$ is non-singular, 
the generalized Nyquist stability-boundary condition in Section ~\ref{Sec.2} can be equivalently expressed as (existing the critical mode $ s_c=j\omega_c$)
\begin{equation}\label{Eq.characteristicBoundary}
\begin{aligned}
\frac{1}{\|\boldsymbol S_{\mathrm m}(s)\|_{\infty}}=0
&\Leftrightarrow
\det(
\mathbb I+
\boldsymbol{\tilde J}_{\mathrm{mNET}}(j\omega_c)
\boldsymbol{\tilde J}_{\mathrm{mCIG}}^{-1}(j\omega_c)
)=0\\
&\Leftrightarrow
\det
\boldsymbol{\tilde C}_{\mathrm{m}}(j\omega_c)
=0
\end{aligned}
\end{equation}
Thus, MCPS reaches its stability boundary 
when its characteristic matrix becomes singular at some frequency
$\omega_c$, $\underline{\sigma}(\boldsymbol{\tilde C}_{\mathrm{m}}(j\omega_c))=0$, or one of the characteristic loci 
$\lambda(\boldsymbol{\tilde C}_{\mathrm{m}}(j\omega))$ crosses the origin.

To establish the relationship between the stability boundaries of the critical single-converter subsystem and the original MCPS, we introduce a nominal MCPS that retains the original grid model and has the converters'  Jacobian transfer matrix
\begin{equation}\label{Eq.nominalCIG}
\boldsymbol{\bar J}_{\mathrm{mCIG}}(s)
:=
\begin{bmatrix}
\bar J_{\mathrm{eqCIG},11}(s)\mathbb I
&
\bar J_{\mathrm{eqCIG},12}(s)\Phi_1
\\
\bar J_{\mathrm{eqCIG},21}(s)\Phi_1^*
&
\bar J_{\mathrm{eqCIG},22}(s)\mathbb I
\end{bmatrix}
\end{equation}
where $\Phi_1 = \operatorname{diag}\{w_{1i}/w_{1i}^{*}\}$. Accordingly, its characteristic matrix is
\begin{equation}\label{Eq.nominalCharacteristic}
\boldsymbol{\bar C}_{\mathrm{m}}(s)
:=
\boldsymbol{\bar J}_{\mathrm{mCIG}}(s)
+\boldsymbol{\tilde J}_{\mathrm{mNET}}(s).
\end{equation} 

\begin{lemma}[]
\label{lem:nominalBoundary} 
Suppose that the critical single-converter subsystem reaches its stability boundary at $s=j\omega_c$, and let $u_c$ be a unit vector satisfying $\bar C_{\mathrm{eq}}(j\omega_c)u_c=\mathbb O.$
Then, the corresponding direction $x_c:=D_1u_c$ satisfies  $\boldsymbol{\bar C}_{\mathrm m}(j\omega_c)x_c=\mathbb O$. 
Therefore, the stability boundary of the critical single-converter subsystem is embedded in that of the nominal MCPS, and the two systems reach the stability boundary at the same critical mode.
\end{lemma}
\noindent\textit{Proof:}
See Appendix~\ref{App:lemNominal}.

The original MCPS is then expressed as a deviation from the nominal
MCPS:$\boldsymbol{\tilde C}_{\mathrm{m}}(s) =
\boldsymbol{\bar C}_{\mathrm{m}}(s)
+\Delta(s)$, where 
$
\Delta(s):=
\boldsymbol{\tilde J}_{\mathrm{mCIG}}(s)
-\boldsymbol{\bar J}_{\mathrm{mCIG}}(s).
$ 
The following proposition quantifies the approximation accuracy of the critical single-converter subsystem in terms of the directional action of $\Delta(s)$. 

\begin{proposition}[Approximation of the original MCPS stability boundary]
\label{prop:originalBoundary}
Under the conditions of Lemma~\ref{lem:nominalBoundary},
if the directional deviation $\varepsilon_c :=\left\|\Delta(j\omega_c)x_c\right\|_2$ is sufficiently small, then $\underline{\sigma}(\boldsymbol{\tilde C}_{\mathrm{m}}(j\omega_c)) \leq \varepsilon_c$ is also small, implying that, at the stability boundary predicted by the critical single-converter subsystem, the original MCPS characteristic matrix is close to singularity and gives  $\underline{\sigma}(\boldsymbol{\tilde C}_{\mathrm{m}}(j\omega_c)) \approx 0, 1/\|\boldsymbol S_{\mathrm m}(s)\|_{\infty}\approx0$. 
\end{proposition}

\noindent\textit{Proof:}
See Appendix~\ref{App:propBoundary}.

Propositions~\ref{prop:preserveNMPzero} and
\ref{prop:originalBoundary} establish the approximation relationship
between the critical single-converter subsystem and the original MCPS.
Since the critical subsystem and the nominal MCPS preserve the critical
NMP zero and its associated performance constraint, the stability-boundary
error is governed mainly by the directional deviation along the
corresponding eigenvector, which is typically small relative to the
overall model mismatch. Consequently, the critical subsystem can
approximate the stability boundary of the original MCPS. }

{\td Moreover, existing aggregated single-machine approaches  typically retain a selected device while compressing the remainder of the system into an equivalent impedance or impedance ratio for frequency-domain stability assessment \cite {liuOscillatoryStabilityCriterion2018,zhuImpedanceMarginRatio2024}. Although this also reduces the original system to an equivalent single-converter subsystem, such impedance compression may weaken the explicit connection between the reduced model and the physical topology and operating condition of the original system. In contrast, the proposed critical subsystem is constructed from the smallest eigenvalue of $\boldsymbol{H}_{\mathrm{eq}}$, which explicitly incorporates the network
topology and operating points of all converter buses.}

The NMP-Z-factor-based analytical framework developed for
SCPS in Section~\ref{Sec.3} can be extended to MCPS through the
critical single-converter subsystem, allowing the small-signal voltage
and synchronization stability margins of MCPS to be assessed using the
same analytical framework. The following stability criteria can then
be established:

\textit{1) Small-Signal Voltage Stability:} 
{\td 
Similar to the SCPS case, the voltage stability boundary of MCPS is given by 
${\rho_z}={\lambda_1}=1+\tau^2$, or equivalently 
$z_{\text{NMP}1}=0$.
This condition corresponds to 
$0=\det(\boldsymbol{J}_{\text{mNET}}(0))$ in \eqref{Eq.25}. Therefore, stable quasi-static or static operation of MCPS requires
\begin{equation}\label{Eq.30}
    z_{\text{NMP}1}>0,{\rho_z}={\lambda_1}>1+\tau^2 .
\end{equation}
Particularly, for $R\ll X,\tau\approx0$, the criterion reduces to $z_{\text{NMP}1}>0,{\rho_z}>1$.
}

\textit{2) Small-Signal Synchronization Stability:} It is governed by the NMP-Z factor $\rho_z$ in weak grids by incorporating key influencing factors such as multi-converter operating points $(P_i,Q_i,U_i,\delta_i)$ and grid parameters $\boldsymbol{B}$. This implies the constraint on control dynamic performance of CIGs via the critical NMP zero, yielding the following criterion
\begin{equation}\label{Eq.31}
    z_\text{NMP1} > z_{1c},{\rho _z} = {\lambda _1} > \rho_{zc}
\end{equation}
where $z_{1c}$ and $\rho_{zc}=\mu_{1c}\mu_{1c}^*$ represent the thresholds of the critical NMP zero and the NMP-Z factor on the small-signal synchronization stability boundary (i.e.,$1/\left \|{S_{{\text{eq}}}}(s,\mu_{1c})\right \| _\infty=0$), linked to control parameters of the given CIGs such as PLL bandwidth.

{\td
\noindent\textbf{Remark 2.}(Revisiting existing gSCR and gOSCR via NMP zeros)\label{Remark:gSCRandNMPZF}
\textit{ 
The gSCR has been widely employed as a grid strength index for assessing small-signal voltage and synchronization stability in MCPS \cite{peng2025measuring,dong2018small}. To account for variations in converter active-power outputs and bus-voltage magnitudes, gOSCR was further introduced as an operation-adapted variant of gSCR \cite{Liu2024gOSCR}. From the perspective of NMP zeros, both gSCR and gOSCR can be recovered as special cases of the proposed NMP-Z factor:
\begin{equation}\label{Eq.32}
\begin{aligned}
    \rho_z &= \underline{\lambda} ( \boldsymbol{P}_N^{-1}\boldsymbol{B}\cdot\boldsymbol{P}_N^{-1}\boldsymbol{B})  = \underline{\lambda} ^2(\boldsymbol{P}_N^{-1}\boldsymbol{B})  = gSCR^2,\\
    \rho_z &= \underline{\lambda} ( \boldsymbol{P}^{-1}\boldsymbol{U}\boldsymbol{B}\boldsymbol{U}\cdot \boldsymbol{P}^{-1}\boldsymbol{U}\boldsymbol{B}\boldsymbol{U})\\
    &=  \underline{\lambda} ^2(\boldsymbol{P}^{-1}\boldsymbol{U}^2\boldsymbol{B})  = gOSCR^2.
\end{aligned}   
\end{equation}
which holds when $\tilde{\boldsymbol{S}} = \boldsymbol{P}_N$ (rated active power, $Q_i = 0$) for gSCR and $\tilde{\boldsymbol{S}} = \boldsymbol{P}$ ($Q_i = 0$) for gOSCR; $\tilde{\boldsymbol{Y}} = \boldsymbol{B}$ (nominal voltages $U_i = 1.0$ p.u. and $\delta_{ij} \approx 0$) for gSCR and $\tilde{\boldsymbol{Y}} = \boldsymbol{U}\boldsymbol{B}\boldsymbol{U}$ ($\delta_{ij} \approx 0$) for gOSCR.}

\textit{ Hence, gSCR corresponds to the rated-condition specialization of gOSCR, whereas gOSCR corresponds to the unity-power-factor and small-angle specialization of the proposed NMP-Z factor. the proposed NMP-Z factor further retains the effects of reactive-power outputs and voltage-angle differences through the complex matrices$\tilde{\boldsymbol{S}}$ and $\tilde{\boldsymbol{Y}}$. 
More importantly, it directly determines the location of the critical NMP zero and thereby provides a control-theoretic foundation for gSCR and gOSCR in assessing both voltage and synchronization stability.
}
}
\subsection{Proposed Unified Stability Margin Assessment Method}
A unified stability margin assessment method can be developed based on the proposed NMP‑Z factor and illustrated in Fig.~9. 

The proposed method leverages the NMP‑Z factor as a static metric, offering several practical advantages. First, it maintains the computational simplicity of traditional SCR/gSCR‑based approaches while providing a more comprehensive representation of multiple converters under various non‑rated operating points. 
Second, all inputs (such as operating point of each converter $(P_i,Q_i,U_i,\delta_i),\; i=1,\dots,n$ and grid susceptance matrix $\boldsymbol{B}$ )
required for computing the NMP-Z factor can be readily available from static phasor measurement unit (PMU) data.
Third, the threshold of NMP-Z factor for MCPS can be determined from a single-converter subsystem. This subsystem is constructed via distributed frequency scanning of individual converters' $J_{\text{CIG}i}(j\omega)$, applicable even to black‑box models. This eliminates the need for full‑system frequency scanning or electromagnetic-transient (EMT) simulation, thereby significantly reducing analysis complexity.
\begin{figure}[H]
\centering
\includegraphics[width=3.0in]{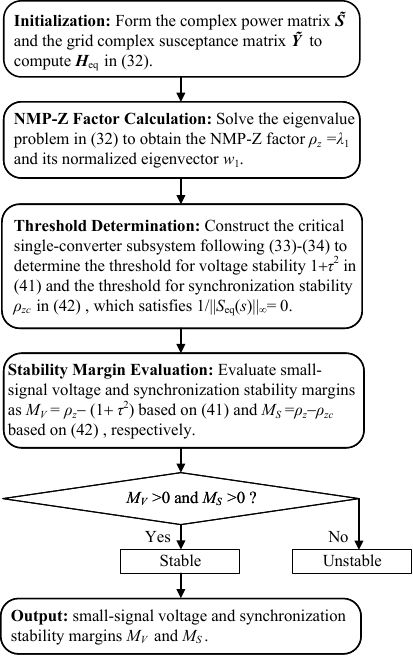}
{\td \caption{Implementation of the Unified Stability Margin Assessment Method. }
}\label{Fig.9-ProposedMethodFlow}
\end{figure}
\vspace{-1em}
\noindent\textbf{Example 2.}(A three-converter power system) 
\textit{
We consider now the three-converter system in Fig.~\ref{Fig.10-Three_Converter} to illustrate the proposed method (parameters of CIGs and grid are given in Appendix B). 
The operating points are listed as Case 1 in Table~II (see the following section). 
{\td Conventional stability assessment for such a multi-converter system typically relies on a high-order state-space or frequency-domain impedance model incorporating all system components, which makes the analysis tedious, especially under varying operating conditions. The proposed approach offers a more convenient alternative and can serve as a rapid pre-screening tool for identifying potential weak-grid voltage and synchronization stability risks before a detailed full-system eigenvalue or broadband impedance analysis is conducted.}
}

\textit{ {\td
Firstly, using the given operating points and grid data, the complex power matrix and grid complex susceptance matrix are obtained as $\boldsymbol{\tilde S} \!= \! \text{diag}\{0.93e^{j0.19},1.09e^{j0.45},1.06e^{j0.44}\}$,
\[
\boldsymbol{\tilde Y} \!= \! \left[ \begin{array}{@{}c@{\,}c@{\,}c@{}}
8.32 & 1.77e^{-j3.12} & 1.02e^{j2.91} \\
1.77e^{j3.12} & 9.97 & 0.37e^{j2.89} \\
1.02e^{-j2.91} & 0.37e^{-j2.89} & 2.46
\end{array} \right].
\] }
}

\textit{ {\td
Next, from $\boldsymbol{\tilde S}$ and $\boldsymbol{\tilde Y}$, we compute $\boldsymbol{H}_{\mathrm{eq}}$. Its smallest eigenvalue and corresponding normalized eigenvector yield the NMP‑Z factor and eigenvector:
$\rho_z = 4.41,w_1=[0.18{e^{j0.17}},0.08{e^{j0.14}},0.98e^{-j0.01}]^T$. The equivalent grid complex Jacobian transfer matrix in \eqref{Eq.28} $\bar J_{\mathrm{eqNET}}(j\omega)$ also can be given. }
}

\textit{ {\td
Then, the frequency responses $J_{\mathrm{CIG}i}(j\omega)$ ($i=1,2,3$) of the individual converters are obtained though individual frequency‑scanning tests. Combining these  frequency responses with $\boldsymbol{w}_1$ gives the equivalent single‑converter complex Jacobian transfer matrix $\bar J_{\mathrm{eqCIG}} (j\omega)$ in \eqref{Eq.29}. Together with $\bar J_{\mathrm{eqNET}}(j\omega)$, the critical single-converter subsystem can be formulated. The threshold of the NMP-Z factor is determined by varying $\mu_1$ until the critical subsystem reaches its stability boundary, i.e., ${1/\left\|S_{\mathrm{eq}}(s,\mu_{1c})\right\|_{\infty}} \approx 0\}$.
This yields $\rho_{zc} = \mu_{1c}\mu_{1c}^{*} = 4.23, {\omega_c}/{2\pi}=25.4~\mathrm{Hz}.$ Moreover, the directional deviation defined in Proposition~\ref{prop:originalBoundary} is calculated as $\varepsilon_c = 0.054 \approx 0$, indicating that the approximation error introduced by the critical subsystem is small. }
}
\begin{figure}[H]
\centering
\includegraphics[width=3.5in]{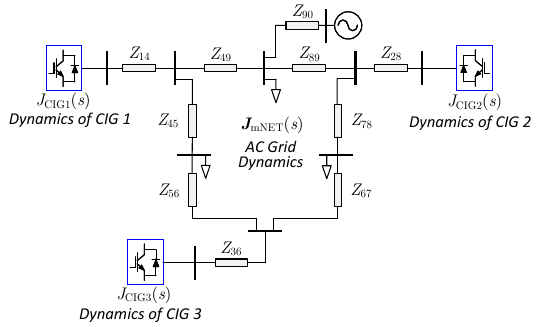}
\caption{A three-converter test system. }\label{Fig.10-Three_Converter}
\end{figure}
\vspace{-0.5cm}
\textit{ {\td
Fig.~11 compares the inverse maximum singular values of the sensitivity functions of the original three-converter system and the critical single-converter subsystem near the stability boundary (i.e., $1/ \overline{\sigma}\{\boldsymbol S_{\mathrm{m}}(j \omega)\}$ and $1/ \overline{\sigma}\{S_\mathrm{eq}(j \omega)\}$ ). 
Since the inverse sensitivity peak (i.e., $1/\|S(s)\|_\infty = \min_\omega 1/ \overline{\sigma}\{{S}(j \omega)\}$) quantifies the stability margin, the close agreement between the two curves verifies Proposition~\ref{prop:originalBoundary}, confirming that the critical subsystem can capture the stability margin of the original MCPS. }
}

\textit{ {\td
Finally, because the grid resistance-to-reactance ratio is $\tau=0.1$, the voltage- and synchronization-stability margins are evaluated as
\[
\begin{aligned}
M_V
&=
\rho_z-(1+\tau^2)
=
4.41-1.01
=
3.40>0,\\
M_S
&=
\rho_z-\rho_{zc}
=
4.41-4.23
=
0.18>0.
\end{aligned}
\]
Both margins are positive, analytically confirming that the three-converter system operates stably under Case~1. }
}
\begin{figure}[H]
\centering
\includegraphics[width=2.8in]{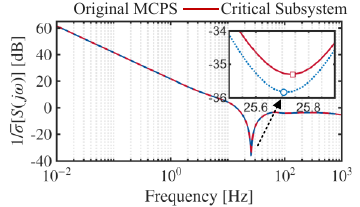}
{\td \caption{ Comparison of inverse maximum singular values of the sensitivity functions for the original three‑converter system and its critical single-converter subsystem near the stability boundary. } }\label{Fig.11-Compare_Subsystem}
\end{figure}

\vspace{-1em}
\section{Simulation Results}
{\td
\subsection{Accuracy Verification}
This subsection verifies the accuracy of the proposed stability assessment method and compares its results with those of existing grid-strength indices. 

The three-converter system in Example 2 is used for the test. Three representative operating conditions are considered, as summarized in Table~II. Case~1 represents stable operation, Case~2 represents the small-signal synchronization instability and Case~3 represents the non-periodic voltage instability.
The grid parameters are identical to those in Example 2 (the grid resistance-reactance ratio is set to $\tau=0.1$), and converter parameters for each Case are provided in Appendix ~\ref{App:paraSystem}.

Table~III summarizes the NMP-Z factor,
critical values, and the corresponding voltage and synchronization
stability margins for the three test cases. For Case~1, both
$M_V$ and $M_S$ remain positive, indicating sufficient voltage and
synchronization stability margins. For Case~2, the synchronization
stability margin becomes negative ($M_S<0$), indicating the
synchronization instability. For Case~3, the voltage stability
margin approaches zero ($M_V\approx0$), indicating that the system
operates close to the voltage stability boundary.
\vspace{-2.5em}
\begin{table}[H]
\centering
\setlength{\tabcolsep}{10pt} 
\renewcommand{\arraystretch}{1.3} 
{\td \caption{Operating Points of Converters in Test Cases (Per Unit)} }
\label{Table.2_operatingpoints_3Converter}
\begin{tabular}{l|ccc}
\hline
& Case 1 & Case 2 & Case 3 \\
\hline
CIG 1 & 
\makecell[c]{$P_1=0.91$\\ $Q_1=0.18$\\ $U_1=0.97$\\ $\delta_1=0.12$} &
\makecell[c]{$P_1=1.01$\\ $Q_1=-0.20$\\ $U_1=0.90$\\ $\delta_1=0.17$} & 
\makecell[c]{$P_1=1.45$\\ $Q_1=-0.42$\\ $U_1=0.83$\\ $\delta_1=0.35$} \\
\hline
CIG 2 & 
\makecell[c]{$P_2=0.99$\\ $Q_2=0.47$\\ $U_2=1.02$\\ $\delta_2=0.11$} & 
\makecell[c]{$P_2=0.99$\\ $Q_2=0.47$\\ $U_2=1.00$\\ $\delta_2=0.12$} & 
\makecell[c]{$P_2=1.39$\\ $Q_2=0.85$\\ $U_2=1.01$\\ $\delta_2=0.21$} \\
\hline
CIG 3 & 
\makecell[c]{$P_3=0.96$\\ $Q_3=0.46$\\ $U_3=0.99$\\ $\delta_3=0.35$} & 
\makecell[c]{$P_3=0.96$\\ $Q_3=0.31$\\ $U_3=0.87$\\ $\delta_3=0.49$} & 
\makecell[c]{$P_3=1.49$\\ $Q_3=0.80$\\ $U_3=0.86$\\ $\delta_3=1.09$} \\
\hline
\end{tabular}
\end{table}
\vspace{-5em}
\begin{table}[H]
\centering
\renewcommand{\arraystretch}{1.2}
\setlength{\tabcolsep}{3pt}
{\td
\caption{Stability Margin Assessment in Test Cases}
}
\label{Table.3_StabilityMargin}
\begin{tabular}{l@{\hspace{5pt}}c@{\hspace{5pt}}c@{\hspace{5pt}}c}
\hline
& Case 1
& Case 2
& Case 3 \\
\hline

NMP-Z Factor
& $\rho_z=4.41$
& $\rho_z=3.02$
& $\rho_z=1.06$ \\

Thresholds
& \makecell[c]{$1+\tau^2=1.01,$\\$\rho_{zc}=4.23$}
& \makecell[c]{$1+\tau^2=1.01,$\\$\rho_{zc}=3.17$}
& \makecell[c]{$1+\tau^2=1.01$} \\

Stability Margins
& \makecell[c]{$M_V=3.40,$\\$M_S=0.18$}
& \makecell[c]{$M_V=2.01,$\\$M_S=-0.15<0$}
& \makecell[c]{$M_V=0.05\approx0$} \\

Criterion
& Stable
& Unstable
& Critically Unstable \\
\hline
\end{tabular}
\end{table}
\vspace{-3em}
Fig.~12 compares the NMP zeros and dominant modes of the original MCPS and the proposed critical subsystem under three operating conditions. As shown in Fig.~12(a), the NMP zeros of the critical subsystem coincide with those of the original system for all cases, indicating that the proposed subsystem successfully preserves the non-minimum-phase characteristics of the original multi-converter system. 
Moreover, Fig.~12(b) compares the dominant poles of the two systems. The dominant modes of the critical subsystem remain close to those of the original system, including both the real unstable
mode and oscillatory modes. Therefore, the proposed critical subsystem can accurately capture the dominant dynamic behavior of the original MCPS.
\begin{figure*}[t]
\centering
  \includegraphics[width=\linewidth]{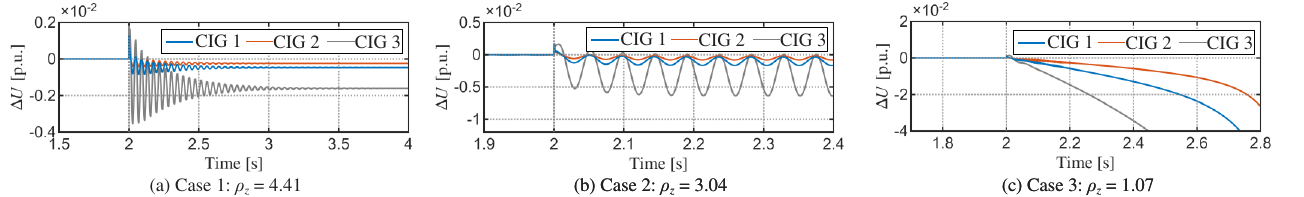}
{\td   \caption{Time-domain voltage responses of CIG 1--3 under three different
  operating conditions: (a) Case 1; (b) Case 2; (c) Case 3.}
  }\label{Fig:VoltageWaves_in_MCPS}
\end{figure*}
\begin{table*}[!t]
\centering
{\td \caption{Assessment Results of Different Grid Strength Indices} }
\label{Table:index_comparison}
\scriptsize
\renewcommand{\arraystretch}{1.05}
\setlength{\tabcolsep}{3.5pt}
\begin{tabular}{@{}lcccc@{}}
\toprule
Indices
& NMP-Z Factor
& ESCR in \cite{IEEEStandardInterconnection}
& gOSCR in \cite{Liu2024gOSCR}
& IMR in \cite{zhuImpedanceMarginRatio2024} \\
\midrule

Expression
&
$\rho_z=\underline{\lambda}( {\boldsymbol{\tilde {S}}^{ - 1}}\boldsymbol{\tilde{Y}}{\boldsymbol{\tilde {S}}^{* - 1}}{\boldsymbol{\tilde {Y}}^*})$ 
&
$\displaystyle
\mathrm{ESCR}^{(i)}
=
\frac{S_{ac,i}}
{P_i+
 \sum_{\substack{j=1\\j\neq i}}^{n}
 P_j Z_{ji}/Z_{ii}}
$
&
$\displaystyle
\underline{\lambda}
(
\operatorname{diag}
({U_i^2/P_i})
\boldsymbol{B}
)
$
&
$\displaystyle
\mathrm{IMR}^{(k)}
=
\frac{|\sigma|}
{\left\|
\operatorname{Res}_{\lambda}^{*}
\boldsymbol{Y}_{kk}^{\mathrm{sys}}(s)
\right\|
\left\|
\boldsymbol{Z}_{Ak}(\lambda)
\right\|}
$
\\
\midrule

Values
&
$\rho_z=4.41$
&
\makecell[c]{
$\mathrm{ESCR}^{(1)}=4.9598,\quad
 \mathrm{ESCR}^{(2)}=6.4025$\\[-1pt]
$\mathrm{ESCR}^{(3)}=2.0052$
}
&
$\mathrm{gOSCR}=2.32$
&
\makecell[c]{
$\mathrm{IMR}^{(1)}=0.225,\quad
 \mathrm{IMR}^{(2)}=0.815$\\[-1pt]
$\mathrm{IMR}^{(3)}=0.007$
}
\\

Thresholds
&
$4.23$
&
$3.0$
&
$2.21$
&
$0$
\\
\midrule

Criterion
&
Stable
&
Unstable
&
Stable
&
Stable
\\
\bottomrule
\end{tabular}
\end{table*}

\begin{figure}[H]
\centering
  \includegraphics[width=2.8in]{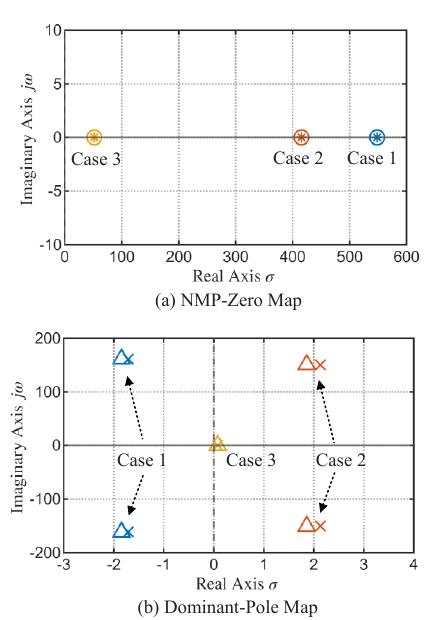}
{\td \caption{Comparison of NMP zeros and dominant modes between the original MCPS and the critical subsystem under different operating conditions: 
(a) NMP-zero map (stars: original MCPS; circles:
critical subsystem);
(b) dominant-mode map (crosses: original MCPS;
triangles: critical subsystem).}}\label{Fig:eigenvaluezero}
\end{figure}

Fig.~13 shows the corresponding time-domain voltage responses of CIGs
following a 0.01 p.u. active-power step applied to all CIGs at 2.0s, which drives the system to operating points of Cases 1–3.
For Case~1, the system exhibits well-damped responses with $\rho_z=4.41$.
As $\rho_z$ decreases to 3.02 in Case~2, 
oscillations appear at 25.3 Hz near the PLL bandwidth, indicating a reduction of synchronization stability margin. For Case~3, where $\rho_z=1.07$, the voltage response diverges, demonstrating voltage instability. These results verify that the proposed NMP-Z-factor-based stability margin can effectively characterize different instability issues.

For further comparison, Table~IV compares the assessment results obtained using ESCR \cite{IEEEStandardInterconnection}, gOSCR (i.e., an operation-adapted variant of gSCR) \cite{Liu2024gOSCR}, IMR \cite{zhuImpedanceMarginRatio2024}, and the proposed NMP-Z factor at the operating point of Case~1. The benchmark eigenvalue analysis and time-domain simulations confirm that the system is stable under this
case. As shown in Table~IV, the ESCR based on empirical methods cannot accurately assess stability. In contrast, gOSCR, IMR, and the proposed NMP-Z factor correctly identify the stable operation.

For gOSCR, Remark~2 has shown that it is a special case of the proposed NMP-Z factor under the ideal conditions of unity power factor of CIGs. Therefore, the NMP-Z factor extends gOSCR to general non-rated operating conditions.

IMR is a mode-oriented impedance margin evaluated at a selected dominant eigenvalue. For the operating condition in Table~IV, the dominant eigenvalue pair $\lambda=-1.71\pm j161.154$ is used to calculate the local IMRs. Its implementation requires the frequency-dependent system admittance $\boldsymbol{Y}_{kk}^{\mathrm{sys}}(s)$, the CIG impedance $\boldsymbol{Z}_{Ak}(s)$, and prior identification of the dominant eigenvalue. Obtaining this information for a large-scale MCPS with multiple heterogeneous CIGs generally increases the modeling and computational burden. By comparison, the proposed NMP-Z factor provides the same correct stability judgment while retaining the computational simplicity of a static grid-strength index.

Moreover, VDSCR also relies on frequency-dependent impedance information. Its formulation is developed for parallel converters with identical control parameters and algorithms in a predominantly inductive grid  \cite{xiaoDesignOrientedSmallSignalStability2025}. These conditions may not be generally satisfied in the practical MCPS.

\subsection{Applicability under Different Grid Conditions}

This subsection examines the applicability of the proposed method under different grid strengths and resistance--reactance ratios. The operating point of Case~1 is adopted as the base condition, while the converter parameters remain unchanged.

\textit{1) Different Grid Strengths:} 
To change the grid strength to generate several cases, a grid-impedance scaling factor $k_{\rm grid}$ is  used to equally increase (or reduce) all line impedances in the power network at the same time, and thus a smaller  $k_{\rm grid}$ represents a stronger grid. 

The calculated NMP-Z factors and the corresponding stability margins are summarized in Table~V.
As the grid becomes stronger, i.e., $k_{\rm grid}$ decreases from 1.0 to 0.5, the NMP-Z factor increases from 4.41 to 4.84. 
The voltage stability margin $M_V$ increases from 3.40 to 3.83. 
The synchronization stability margin $M_S$ increases from 0.18 to 0.32, although the threshold $\rho_{zc}$ increases slightly from 4.23 to 4.52 due to the variation of the critical subsystem with grid conditions.
\begin{figure}[H]
\centering
\includegraphics[width=2.8in]{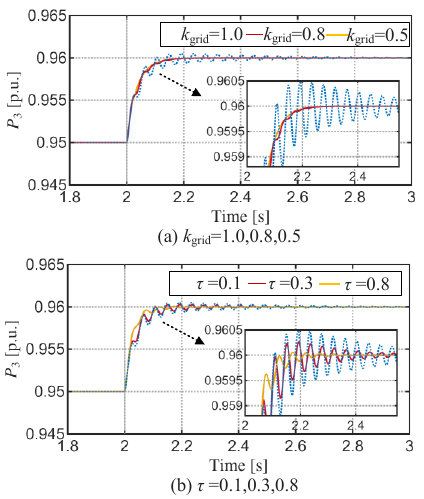}
{\td \caption{Time-domain active power responses of CIG 3 in the test MCPS under different grid conditions. (a) different grid strength: $k_{\rm grid}=1.0,0.8,0.5$ ; (b) different resistance--reactance ratios: $\tau = 0.1,0.3,0.8$.} }
\label{Fig:powewavegrid}
\end{figure}

\begin{table}[H]
\centering
\renewcommand{\arraystretch}{1.2}
\setlength{\tabcolsep}{3pt}
{\td \caption{Stability Assessment under Different Grid Strengths} }
\label{Table:grid_strength}
\begin{tabular}{l@{\hspace{5pt}}c@{\hspace{5pt}}c@{\hspace{5pt}}c}
\hline
& $k_{\rm grid}=1.0$
& $k_{\rm grid}=0.8$
& $k_{\rm grid}=0.5$ \\
\hline

NMP-Z Factor
& $\rho_z=4.41$
& $\rho_z=4.67$
& $\rho_z=4.84$ \\

Thresholds 
& \makecell[c]{$1+\tau^2=1.01,$\\$\rho_{zc}=4.23$}
& \makecell[c]{$1+\tau^2=1.01,$\\$\rho_{zc}=4.39$}
& \makecell[c]{$1+\tau^2=1.01,$\\$\rho_{zc}=4.52$} \\

Stability Margins
& \makecell[c]{$M_V=3.40,$\\$M_S=0.18$}
& \makecell[c]{$M_V=3.66,$\\$M_S=0.28$}
& \makecell[c]{$M_V=3.83,$\\$M_S=0.32$} \\

Criterion
& Stable
& Stable
& Stable \\
\hline
\end{tabular}
\end{table}
\begin{table}[H]
\centering
\renewcommand{\arraystretch}{1.2}
\setlength{\tabcolsep}{3pt}
{\td \caption{Stability Assessment under Different Resistance--Reactance Ratios} }
\label{Table:impedance_ratio}
\begin{tabular}{l@{\hspace{5pt}}c@{\hspace{5pt}}c@{\hspace{5pt}}c}
\hline
& $\tau=0.1$
& $\tau=0.3$
& $\tau=0.8$ \\
\hline

NMP-Z Factor
& $\rho_z=4.41$
& $\rho_z=5.63$
& $\rho_z=8.39$ \\

Thresholds 
& \makecell[c]{$1+\tau^2=1.01,$\\$\rho_{zc}=4.23$}
& \makecell[c]{$1+\tau^2=1.09,$\\$\rho_{zc}=5.08$}
& \makecell[c]{$1+\tau^2=1.64,$\\$\rho_{zc}=6.94$} \\

Stability Margins
& \makecell[c]{$M_V=3.40,$\\$M_S=0.18$}
& \makecell[c]{$M_V=4.54,$\\$M_S=0.55$}
& \makecell[c]{$M_V=6.75,$\\$M_S=1.45$} \\

Criterion
& Stable
& Stable
& Stable \\
\hline
\end{tabular}
\end{table}
}
{\td Fig.~14(a) compares the active-power responses of the CIG~3 for different grid strengths. As shown in Fig.~14(a), reducing
$k_{\rm grid}$ leads to faster damping and smaller oscillation amplitudes,
which agrees with the increasing stability margins predicted by the
proposed method. This agreement confirms the applicability of the proposed NMP-Z factor to stability assessment under varying grid strengths.

\textit{2) Different Resistance--Reactance Ratios:} The influence of the network resistance is further investigated by setting the resistance--reactance ratio to $\tau=0.1$, $0.3$, and $0.8$.
As $\tau$ increases, the grid changes from a predominantly inductive network to a more resistive one.
For each value of $\tau$, the power-flow operating point and the corresponding network matrix (i.e.,$\tilde{\boldsymbol S}$ and $\tilde{\boldsymbol Y}$ in \eqref{Eq.27}) are recalculated, resulting in different NMP-Z factors. 

The calculated NMP-Z factors and the corresponding stability margins are summarized in Table~VI. 
As $\tau$ increases from $0.1$ to $0.8$, the NMP-Z factor $\rho_z$ increases from $4.41$ to $8.39$. Meanwhile, the voltage-stability threshold $1+\tau^2$ increases from $1.01$ to $1.64$, and the synchronization-stability threshold $\rho_{zc}$ increases from $4.23$ to $6.94$.
Accordingly, the voltage stability margin
$M_V$ increases from $3.40$ to $6.75$, while the synchronization stability margin
$M_S$ increases from $0.18$ to $1.45$.
Both margins remain positive for all three cases, predicting stable operation.

Fig.~14(b) compares the active-power responses of CIG~3 under different resistance--reactance ratios. 
As $\tau$ increases, the oscillation amplitude decreases and the response becomes more rapidly
damped. This trend is consistent with the increases in both $M_V$ and $M_S$ reported in Table~VI. 
The agreement confirms that the proposed method can  reflect the influence of the resistance--reactance ratio on system dynamic performance.

\subsection{Applicability under Different Converter Control Schemes}
This subsection verifies the applicability of the proposed method under different converter control schemes and parameters. In contrast to Cases~1--3, where all CIGs adopt the APC--RPC control scheme, CIG~2 is configured with a DVC--RPC control scheme (PI parameters of DVC:{1,8} and PI parameters of RPC:{0.5,40}), while CIGs~1 and~3 retain APC--RPC control with different parameters. The PLL bandwidth of CIG~1 is fixed at $5$~Hz. Based on the operating condition of Case~1, three
additional cases are generated by setting the PLL bandwidth of CIG~3 to $20$, $15$, and $5$~Hz, respectively.

As summarized in Table~VII, reducing the PLL bandwidth of CIG~3 from $20$ to $5$~Hz decreases the synchronization-stability threshold $\rho_{zc}$ from $4.23$ to $4.06$, and thus the synchronization stability margin
$M_S$ increases from $0.18$ to $0.35$. Since the grid parameters and operating point remain unchanged,
$\rho_z=4.41$, the voltage-stability threshold
$1+\tau^2=1.01$, and the voltage stability margin
$M_V=3.40$ are identical in all three cases. All three cases
remain stable, while the increasing $M_S$ indicates progressively improved damping performance.

Fig.~15 compares the voltage responses of CIG~3 under different PLL bandwidths. As the PLL bandwidth decreases from $20$~Hz to $15$ and $5$~Hz, the
oscillations become progressively smaller and more rapidly damped. This trend agrees with the increase in the synchronization stability margin $M_S$ reported in Table~VII.

The results also demonstrate the different roles of the NMP-Z factor and its critical value. The unchanged $\rho_z$ reflects the fixed grid and operating condition, whereas the variation of $\rho_{zc}$ captures the influence of the converter control dynamics. The agreement between the calculated margins and the time-domain responses confirms that the proposed method is applicable to MCPSs with heterogeneous control schemes and parameters.}

\begin{table}[H]
\centering
\renewcommand{\arraystretch}{1.2}
\setlength{\tabcolsep}{3pt}
{\td \caption{Stability Assessment under Different PLL Bandwidths} }
\label{Table:control_schemes}
\begin{tabular}{l@{\hspace{5pt}}c@{\hspace{5pt}}c@{\hspace{5pt}}c}
\hline
& $\omega_{\rm PLL}=20$ Hz
& $\omega_{\rm PLL}=15$ Hz
& $\omega_{\rm PLL}=5$ Hz \\
\hline

NMP-Z Factor
& $\rho_z=4.41$
& $\rho_z=4.41$
& $\rho_z=4.41$ \\

Thresholds
& \makecell[c]{$1+\tau^2=1.01,$\\$\rho_{zc}=4.23$}
& \makecell[c]{$1+\tau^2=1.01,$\\$\rho_{zc}=4.15$}
& \makecell[c]{$1+\tau^2=1.01,$\\$\rho_{zc}=4.06$} \\

Stability Margins
& \makecell[c]{$M_V=3.40,$\\$M_S=0.18$}
& \makecell[c]{$M_V=3.40,$\\$M_S=0.26$}
& \makecell[c]{$M_V=3.40,$\\$M_S=0.35$} \\

Criterion
& Stable
& Stable
& Stable \\
\hline
\end{tabular}
\end{table}

\begin{figure}[H]
\centering
\includegraphics[width=3.0in]{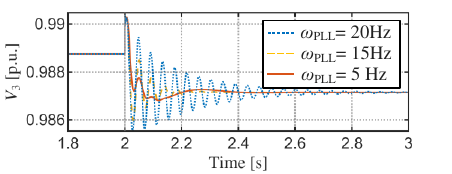}
{\td \caption{Time-domain voltage responses of CIG 3 in the test MCPS under different PLL bandwidths.} }
\label{Fig:voltagewavePLL}
\end{figure}

\section{Conclusion}
This paper reinterpreted the small-signal voltage and synchronization stability of converter‑integrated weak grids through the insights of NMP zeros.
By Jacobian transfer matrix modeling, we derived an analytical expression for NMP zeros and introduced a novel stability metric, termed NMP‑Z factor.
{\td The proposed metric provides a common assessment framework for voltage stability (tied to the NMP zero at the origin) and synchronization stability (tied to low-frequency or sideband NMP zeros that constrained the PLL dynamic performance).
The widely-used SCR is proven to be special cases of the NMP-Z factor under rated conditions, thereby providing a control-theoretic foundation for their practical wisdom.
Inspired by this insight, we developed an NMP-Z-factor-based stability assessment method that retained the simplicity of SCR‑based approaches. It required only the static metric and frequency-response models of individual converters, making it applicable to multi‑converter systems under various operating points.
Our results provided a new perspective for quantifying weak grid issues by rigorously rethinking the fundamental limitations for converters.}
Future work will extend the framework to the optimal operation of converter-based power systems with multiple stability constraints.

\appendix
\renewcommand{\thesection}{\Alph{section}}
{\td
\subsection{Proof of Lemma~\ref{lem:mimo_nmp_bound}}\label{App:lemMIMO}
From $w_z^H L(z)=0$ and $w_z^H S(z)=w_z^H$, we have the scalar directional sensitivity $s_z(s)=w_z^H S(s)w_z$ satisfies $s_z(z)=1$. Since $\|w_z\|_2=1$, the low-frequency condition in \eqref{Eq:low_freq_spec_mimo} gives $|s_z(j\omega)|\leq\overline{\sigma}(S(j\omega))\leq M_L$ for $|\omega|\leq\omega_L$. Applying the SISO bound in Lemma~\ref{lem:siso_nmp_bound} to $s_z(s)$ yields  \eqref{Eq:mimo_nmp_lower_bound}.
}

\subsection{Derivation of \eqref{Eq.6} and \eqref{Eq.20}}\label{App:derivGrid}
Since the ac grid in SCPS is the special case of ac grid in MCPS when $n=1$, here we only give the derivation of $\boldsymbol{J}_\text{mNET}(s)$ and  ${J}_\text{NET}(s)$ is its special case in SCPS.

{\td Consider the dynamics of line inductance for any $(i,j)\in\cal E$ can be represented as ($\tau$ is the identical $R_{ij}/X_{ij}$ ratio of all lines)
\begin{equation}\label{Eq.33}
\left[ \begin{array}{l}
{U_{di}-U_{dj}}\\
{U_{qi}-U_{qj}}
\end{array} \right] \!\!
= \!\!
\left[ {\begin{array}{*{20}{c}}
{\tau X_{ij} + s{L_{ij}}}&{ -X_{ij}}\\
{X_{ij}}&{\tau X_{ij} +s{L_{ij}}}
\end{array}} \right]\left[ \begin{array}{l}
{I_{d,ij}}\\
{I_{q,ij}}
\end{array} \right]
\end{equation}
where $I$ represents the current on branch $ij$, subscribe "\textit{d}" and "\textit{q}" represent the \textit{d}-axis an \textit{q}-axis components in synchronous rotating reference \textit{dq}-frame respectively.

The relationship between the power and current is expressed 
\begin{equation}\label{Eq.34}
\left\{ \begin{array}{l}
{P_{ij}} = {U_{{{d}}i}}{I_{{{d,}}ij}} + {U_{{{q}}i}}{I_{{{q,}}ij}}\\
{Q_{ij}} =  - {U_{{{d}}i}}{I_{{{q,}}ij}} + {U_{{{q}}i}}{I_{{{d,}}ij}}
\end{array} \right.
\end{equation}
\begin{equation}\label{Eq.35}
\left[ \begin{array}{l}
{P_i}\\
{Q_i}
\end{array} \right] - \left[ \begin{array}{l}
\sum\nolimits_{j = 1,j \ne i} {{P_{ij}}} \\
\sum\nolimits_{j = 1,j \ne i} {{Q_{ij}}} 
\end{array} \right] = \mathbb{O}
\end{equation}
where $P_{ij}$ and $Q_{ij}$ represent the active power and reactive power on the branch $ij$ ,respectively; $P_{ij}$ and $Q_{ij}$  satisfy the power flow equations as (\ref{Eq.35});  $P_{i}$ and $Q_{i}$  represent the active power and reactive power injection of the $i$-th converter,$i=\{1,...,n\}$.

Then, the transformation between the voltage magnitude-synchronization polar reference frame and the voltage \textit{dq}-frame is expressed as
\begin{equation}\label{Eq.36}
    \left\{ \begin{array}{l}
{U_{{{d}}i}} = {U_i}\cos {\delta _i}\\
{U_{{{q}}i}} = {U_i}\sin {\delta _i}
\end{array} \right.
\end{equation}

Linearizing (\ref{Eq.33})-(\ref{Eq.36}) and the power-voltage dynamics at the $i$-th converter bus is expressed as 
\begin{equation}\label{Eq.37}
\hspace*{-0.4em}  
\setlength{\arraycolsep}{1.2pt}
\medmuskip=1mu
\thickmuskip=1mu 
\begin{aligned}
\left[ \begin{array}{c}
\Delta P_i\\[2pt]
\Delta Q_i
\end{array} \right] &= 
\Bigg(\!\left[ \begin{array}{cc}
P_i & -Q_i\\[2pt]
Q_i & P_i
\end{array} \right] + U_i^2 B_{ii}
\left[ \begin{array}{cc}
\beta(s) & \alpha(s)\\[2pt]
\alpha(s) & -\beta(s)
\end{array} \right]\!\Bigg)
\left[ \begin{array}{c}
U_i^{-1}\Delta U_i\\[2pt]
\Delta \delta_i
\end{array} \right] \\
&\quad+ \sum_{\substack{j=1\\j\neq i}} 
\Bigg(B_{ij}U_iU_j\cos\delta_{ij}
\left[ \begin{array}{cc}
\beta(s) & \alpha(s)\\[2pt]
\alpha(s) & -\beta(s)
\end{array} \right] \\
&\quad\phantom{+}+ B_{ij}U_iU_j\sin\delta_{ij}
\left[ \begin{array}{cc}
-\alpha(s) & \beta(s)\\[2pt]
\beta(s) & \alpha(s)
\end{array} \right]\!\Bigg)
\left[ \begin{array}{c}
U_j^{-1}\Delta U_j\\[2pt]
\Delta \delta_j
\end{array} \right]
\end{aligned}
\end{equation}

Aggregating the dynamics of $n$ converters yields the Jacobian transfer matrix of the ac grid in MCPS as \eqref{Eq.20}.
}

\subsection{System Parameters}\label{App:paraSystem}
The main parameters of the single-converter system in Example 1 are: $L_\text{f} = 0.05$p.u., $C_\text{f} = 0.05$p.u., $R_\text{f} = 0.02$p.u. for filter; $L_\text{g} = 0.5$p.u., $R_\text{g} = 0.01$p.u. for grid transmission Line; DC Capacitor: 0.038 p.u.; PI parameters of current controller:\{0.3,10\}, time constant for voltage forward: 0.002, PI parameters of APC:\{1,20\}, PI parameters of RPC:\{1,20\}, PLL bandwidth: 15 Hz (unless otherwise specified).

{\td The main parameters of the three-converter system in Example 2 are as follows. The base value for power is 100 MVA and for frequency it is 50 Hz. Impedance  ($Z_{ij}=R_{ij}+jX_{ij}$) of transmission lines(p.u.): $Z_{14}=0.0029+j0.0285$,$Z_{45}=0.0228+j0.2280$,$Z_{49}=0.0086+j0.0855$,$Z_{56}=0.0114+j0.1140$,$Z_{36}=0.0200+j0.1995$,$Z_{28}=0.0057+j0.0570$,$Z_{67}=0.002+j0.2280$,$Z_{78}=0.0399+j0.3990$,$Z_{89}=0.0029+j0.0285$,$Z_{90}=0.0029+j0.0285$. Loads are modeled as a constant impedance load with same parameters:$R_\text{load}=0.78$ p.u., $X_\text{load}=0.6$ p.u.; Control parameters of three converters are:  $L_\text{f} = 0.05$ p.u.,$C_\text{f} = 0.05$ p.u., $R_\text{f} = 0.02$ p.u. for filter; PI parameters of current controller:\{0.3,10\}, time constant for voltage forward: 0.001, PI parameters of APC:\{0.5,40\}, PI parameters of RPC:\{0.5,40\}, PLL bandwidth: 20 Hz in Case 1-2 and 2 Hz in Case 3  (unless otherwise specified).
}

\subsection{Proof of Lemma~\ref{lem:eigenvalue_Heq}}\label{App:lemHeq}
The complex power matrix and the grid complex susceptance matrix in (\ref{Eq.23}) can be rewritten as 
\begin{equation}\label{Eq.38}
    \boldsymbol{\tilde S} =  \boldsymbol{S}\text{diag}\{e^{j{\varphi_i}}\} ,\boldsymbol{\tilde Y} = \text{diag}\{ e^{j{\delta _i}}\} \boldsymbol{UBU}\text{diag}\{ e^{-j{\delta _i}}\}
\end{equation}

By combing \eqref{Eq.38},  ${\boldsymbol{H}_\text{eq}} \!=\! {\boldsymbol{\tilde {S}}^{ - 1}}\boldsymbol{\tilde{Y}}{\boldsymbol{\tilde {S}}^{* - 1}}{\boldsymbol{\tilde {Y}}^*}$ can be rewritten
\begin{equation}\label{Eq.39}
\hspace*{-0.7em}  
    {\boldsymbol{H}_\text{eq}}\!=\! \text{diag}\{e^{j{(\delta_i-\varphi_i)}}\}(\boldsymbol{US^{-1}BS^{-1}U})\text{diag}\{e^{j{(\varphi_i-\delta_i)}}\}
\end{equation}
In \eqref{Eq.39}, due to  voltage magnitude and apparent power of converters $U_i>0$ and $S_i>0$, $\boldsymbol{U}$ and $\boldsymbol{S}$ are positive definite. Since $\boldsymbol{B}$  is a grounded Laplacian matrix for grid, $\boldsymbol{B}$ is also positive definite \cite{dorflerKronReductionGraphs2013b}. Thus, $\boldsymbol{UBU}$ and $\boldsymbol{US^{-1}BS^{-1}U}$  are both positive definite. Furthermore, notice that $\text{diag}\{e^{j{(\delta_i-\varphi_i)}}\}$ and $\text{diag}\{e^{j{(\varphi_i-\delta_i)}}\}$are all unitary matrices, their operation on the positive definite matrices and  in  does not change their eigenvalues \cite{zhang2011matrix}. Thus, $\boldsymbol{H}_\text{eq}$ can be proved to be positive definite, and thus its eigenvalues are isolated and positive. 
\hfill$\blacksquare$

{\td
\subsection{Proof of Proposition~\ref{prop:preserveNMPzero}}\label{App:propNMP}

First, let $\boldsymbol{N} = \tilde{\boldsymbol{S}}^{-1}\tilde{\boldsymbol{Y}}$. 
The normalized eigenvector $\boldsymbol{w}_1$ satisfies $\boldsymbol{H}_\text{eq}{w_1} = \boldsymbol{N}{\boldsymbol{N}^*}{w_1} = {\lambda _1}w_1,w_1^H{w_1} = 1$. 
Taking its complex conjugate gives
$\boldsymbol{N}^{*}\boldsymbol{N}w_1^{*}=\lambda_1w_1^{*}$.
Moreover,
$\boldsymbol{N}^{*} \boldsymbol{N}(\boldsymbol{N}^{*}w_1)
=\lambda_1(\boldsymbol{N}^{*}w_1).
$
Thus, both $w_1^*$ and $\boldsymbol{N}^*w_1$ belong to the eigenspace
of $\boldsymbol{N}^*\boldsymbol{N}$ associated with $\lambda_1$.
Since $\lambda_1$ is simple, there exists a scalar $\mu_1$ such that
$
\boldsymbol{N}^{*}w_1=\mu_1^{*}w_1^{*},\boldsymbol{N}w_1^{*}=\mu_1w_1
$,and ${\mu_1} = w_1^H\boldsymbol{N}w_1^*$.
Substituting these relations gives $\boldsymbol{N}{\boldsymbol{N}^*}{w_1} = \mu _1^*\boldsymbol{N}w_1^* = \mu _1^*{\mu _1}{w_1}$ and $({\lambda _1} - \mu _1^*{\mu _1}){w_1} = \mathbb{O}$. Thus, ${\rho _z} = {\lambda _1} = {\mu _1}\mu _1^*$.
\hfill$\blacksquare$
}

{\td
\subsection{Proof of Lemma~\ref{lem:nominalBoundary}}\label{App:lemNominal}
From the definition of $\boldsymbol{\bar J}_{\mathrm{mCIG}}(s)$, we have 
\begin{equation}\label{Eq:JeqCIG_invar}
    \bar{\boldsymbol J}_{\mathrm{mCIG}}(s)D_1=D_1\bar J_{\mathrm{eqCIG}}(s).
\end{equation}

Moreover, the smallest eigenvalue $\lambda_1$ satisfies $\boldsymbol Nw_1^*=\mu_1w_1,
\boldsymbol N^*w_1=\mu_1^*w_1^*$, which gives 
\begin{equation}\label{Eq:JeqNET_invar}
   \boldsymbol{\tilde J}_{\mathrm{mNET}}(s)D_1  = D_1\bar J_{\mathrm{eqNET}}(s)
\end{equation}

Combing \eqref{Eq:JeqCIG_invar} and \eqref{Eq:JeqNET_invar} yields
\begin{equation}
\begin{aligned}
\boldsymbol{\bar C}_{\mathrm{m}}(s)D_1
&=
(\bar{\boldsymbol J}_{\mathrm{mCIG}}(s)+\boldsymbol{\tilde J}_{\mathrm{mNET}}(s))D_1\\
&=
D_1
(\bar J_{\mathrm{eqCIG}}(s)+\bar J_{\mathrm{eqNET}}(s))\\
&=
D_1\bar C_{\mathrm{eq}}(s).
\end{aligned}
\end{equation}
Hence, according to Theorem 5.1 in \cite{stewart1990matrix}, $\mathcal D_1=\operatorname{col}(D_1)$ is the invariant subspace under $\bar{\boldsymbol C}_{\mathrm m}(s)$. If $\bar C_{\mathrm{eq}}(j\omega_c)u_c=\mathbb O$, 
then $\boldsymbol{\bar C}_{\mathrm{m}}(j\omega_c)x_c 
=\boldsymbol{\bar C}_{\mathrm{m}}(j\omega_c)D_1u_c
=D_1\bar C_{\mathrm{eq}}(j\omega_c)u_c=\mathbb O$.
Since $D_1^HD_1=\mathbb I$ and $\|u_c\|_2=1$, one also has
$\|x_c\|_2=1$. Thus, $\boldsymbol{\bar C}_{\mathrm{m}}(j\omega_c)$ is singular whenever
$\bar C_{\mathrm{eq}}(j\omega_c)$ is singular, proving the stated equivalence along the critical subspace. \hfill$\blacksquare$
}

{\td
\subsection{Proof of Proposition~\ref{prop:originalBoundary}}\label{App:propBoundary}
From Lemma~\ref{lem:nominalBoundary}, $\boldsymbol{\bar C}_{\mathrm{m}}(j\omega_c)x_c=\mathbb O, \|x_c\|_2=1$. Using \eqref{Eq.characteristicMatrices} obtains $\boldsymbol{\tilde C}_{\mathrm m}(j\omega_c)x_c 
=({\bar C}_{\mathrm m}(j\omega_c)+{\Delta}(j\omega_c))x_c=\Delta(j\omega_c)x_c$. 
The variational characterization of the minimum singular value gives 
\begin{equation}\label{Eq:singularValueResidual}
 \begin{aligned}
 \underline{\sigma}(\boldsymbol{\tilde C}_{\mathrm{m}}(j\omega_c)) 
&= \min_{\|x\|_2=1}
\left\|
\boldsymbol{\tilde C}_{\mathrm{m}}(j\omega_c)x
\right\|_2 \\
&\leq
\left\|
\tilde{\boldsymbol C}_{\mathrm m}(j\omega_c)x_c
\right\|_2 =
\left\|
\Delta(j\omega_c)x_c
\right\|_2=
\varepsilon_c
\end{aligned}
\end{equation}
Therefore, a sufficiently small normalized directional deviation $\varepsilon_c \approx0$
implies that the original MCPS has approximately singular characteristic matrix at the stability boundary predicted by the critical single-converter subsystem. Combing \eqref{Eq:singularValueResidual} and \eqref{Eq.characteristicBoundary} yields $\underline{\sigma}(\boldsymbol{\tilde C}_{\mathrm{m}}(j\omega_c)) \approx 0$ and $ 1/\|\boldsymbol S_{\mathrm m}(s)\|_{\infty}\approx0$.  
\hfill$\blacksquare$
}

\bibliographystyle{IEEEtran}
\bibliography{IEEEabrv,Bibliography}

@IEEEtranBSTCTL{IEEEexample:BSTcontrol,
    CTLuse_article_number = "yes",
    CTLuse_paper = "yes",
    CTLuse_forced_etal = "yes",        
    CTLmax_names_forced_etal = "3",   
    CTLnames_show_etal = "3",          
    CTLuse_alt_spacing = "yes",
    CTLalt_stretch_factor = "4",
    CTLdash_repeated_names = "yes",
    CTLname_format_string = "{f.~}{vv~}{ll}{, jj}",
    CTLname_latex_cmd = "",
    CTLname_url_prefix = "[Online]. Available:"
}

@article{havre1998effect,
  title={Effect of RHP zeros and poles on the sensitivity functions in multivariable systems},
  author={Havre, Kjetil and Skogestad, Sigurd},
  journal={Journal of Process Control},
  volume={8},
  number={3},
  pages={155--164},
  year={1998},
  publisher={Elsevier}
}

@book{skogestad2005multivariable,
  title={Multivariable feedback control: analysis and design},
  author={Skogestad, Sigurd and Postlethwaite, Ian},
  year={2005},
  publisher={john Wiley \& sons}
}

@article{chen1995sensitivity,
  title={Sensitivity integral relations and design trade-offs in linear multivariable feedback systems},
  author={Chen, Jie},
  journal={IEEE Transactions on Automatic Control},
  volume={40},
  number={10},
  pages={1700--1716},
  year={1995},
  publisher={IEEE}
}

@article{lund2022operating,
  title={Operating wind power plants under weak grid conditions considering voltage stability constraints},
  author={Lund, Torsten and Wu, Heng and Soltani, Hamid and Nielsen, John Godsk and Andersen, Gert Karmisholt and Wang, Xiongfei},
  journal={IEEE Transactions on Power Electronics},
  volume={37},
  number={12},
  pages={15482--15492},
  year={2022},
  publisher={IEEE}
}

@book{zhang2006schur,
  title={The Schur complement and its applications},
  author={Zhang, Fuzhen},
  volume={4},
  year={2006},
  publisher={Springer Science \& Business Media}
}

@book{machowski2020power,
  title={Power system dynamics: stability and control},
  author={Machowski, Jan and Lubosny, Zbigniew and Bialek, Janusz W and Bumby, James R},
  year={2020},
  publisher={John Wiley \& Sons}
}

@article{zhang2010analysis,
  title={Analysis of stability limitations of a VSC-HVDC link using power-synchronization control},
  author={Zhang, Lidong and Nee, Hans-Peter and Harnefors, Lennart},
  journal={IEEE Transactions on Power Systems},
  volume={26},
  number={3},
  pages={1326--1337},
  year={2010},
  publisher={IEEE}
}

@article{yang2019comparison,
  title={Comparison of impedance model and amplitude--phase model for power-electronics-based power system},
  author={Yang, Ziqian and Mei, Cong and Cheng, Shijie and Zhan, Meng},
  journal={IEEE Journal of Emerging and Selected Topics in Power Electronics},
  volume={8},
  number={3},
  pages={2546--2558},
  year={2019},
  publisher={IEEE}
}

@article{he2024complex,
  title={Complex-frequency synchronization of converter-based power systems},
  author={He, Xiuqiang and H{\"a}berle, Verena and D{\"o}rfler, Florian},
  journal={IEEE Transactions on Control of Network Systems},
  volume={12},
  number={1},
  pages={787--799},
  year={2024},
  publisher={IEEE}
}

@article{wang2018harmonic,
  title={Harmonic stability in power electronic-based power systems: Concept, modeling, and analysis},
  author={Wang, Xiongfei and Blaabjerg, Frede},
  journal={IEEE Transactions on Smart Grid},
  volume={10},
  number={3},
  pages={2858--2870},
  year={2018},
  publisher={IEEE}
}

@article{wangSmallSignalStabilityAnalysis2018,
  title = {Small-{{Signal Stability Analysis}} of {{Inverter-Fed Power Systems Using Component Connection Method}}},
  author = {Wang, Yanbo and Wang, Xiongfei and Chen, Zhe and Blaabjerg, Frede},
  year = {2018},
  month = {sep},
  journal = {IEEE Transactions on Smart Grid},
  volume = {9},
  number = {5},
  pages = {5301--5310},
  issn = {1949-3061},
  doi = {10.1109/TSG.2017.2686841},
  langid = {american}
}

@article{harneforsModelingThreePhaseDynamic2007,
  title = {Modeling of {{Three-Phase Dynamic Systems Using Complex Transfer Functions}} and {{Transfer Matrices}}},
  author = {Harnefors, Lennart},
  year = 2007,
  month = aug,
  journal = {IEEE Transactions on Industrial Electronics},
  volume = {54},
  number = {4},
  pages = {2239--2248},
  issn = {1557-9948},
  doi = {10.1109/TIE.2007.894769},
  urldate = {2024-12-27},
  langid = {american}
}

@article{huangSmallDisturbanceVoltageStability2020,
  title = {Small-{{Disturbance Voltage Stability}} of {{Power Systems}}: {{Dependence}} on {{Network Structure}}},
  shorttitle = {Small-{{Disturbance Voltage Stability}} of {{Power Systems}}},
  author = {Huang, Wanjun and Hill, David J. and Zhang, Xinran},
  year = {2020},
  month = {jul},
  journal = {IEEE Transactions on Power Systems},
  volume = {35},
  number = {4},
  pages = {2609--2618},
  issn = {1558-0679},
  doi = {10.1109/TPWRS.2019.2962555},
  langid = {american}
}

@article{wang2020grid,
  title={Grid-synchronization stability of converter-based resources—An overview},
  author={Wang, Xiongfei and Taul, Mads Graungaard and Wu, Heng and Liao, Yicheng and Blaabjerg, Frede and Harnefors, Lennart},
  journal={IEEE Open Journal of Industry Applications},
  volume={1},
  pages={115--134},
  year={2020},
  publisher={IEEE}
}

@article{cheng2022real,
  title={Real-world subsynchronous oscillation events in power grids with high penetrations of inverter-based resources},
  author={Cheng, Yunzhi and Fan, Lingling and Rose, Jonathan and Huang, Shun-Hsien and Schmall, John and Wang, Xiaoyu and Xie, Xiaorong and Shair, Jan and Ramamurthy, Jayanth R and Modi, Nilesh and others},
  journal={IEEE Transactions on Power Systems},
  volume={38},
  number={1},
  pages={316--330},
  year={2022},
  publisher={IEEE}
}

@article{gu2022power,
  title={Power system stability with a high penetration of inverter-based resources},
  author={Gu, Yunjie and Green, Timothy C},
  journal={Proceedings of the IEEE},
  volume={111},
  number={7},
  pages={832--853},
  year={2022},
  publisher={IEEE}
}

@article{dong2018small,
  title={Small signal stability analysis of multi-infeed power electronic systems based on grid strength assessment},
  author={Dong, Wei and Xin, Huanhai and Wu, Di and Huang, Linbin},
  journal={IEEE transactions on Power Systems},
  volume={34},
  number={2},
  pages={1393--1403},
  year={2018},
  publisher={IEEE}
}

@article{peng2025measuring,
  title={Measuring Short-Term Voltage Stability of Power Systems Dominated by Inverter-Based Resources Part I: System-Wise Generalized Voltage Damping Index},
  author={Peng, Xiaoyu and Liu, Feng and Yang, Peng and Yu, Peixin and Luo, Kui and Wang, Zhaojian},
  journal={Journal of Modern Power Systems and Clean Energy},
  year={2025},
  publisher={SGEPRI}
}

@book{zhang2011matrix,
  title={Matrix theory: basic results and techniques},
  author={Zhang, Fuzhen},
  year={2011},
  publisher={Springer Science \& Business Media}
}

@article{dorflerKronReductionGraphs2013b,
  title = {Kron {{Reduction}} of {{Graphs With Applications}} to {{Electrical Networks}}},
  author = {Dorfler, Florian and Bullo, Francesco},
  year = {2013},
  journal = {IEEE Transactions on Circuits and Systems I: Regular Papers},
  volume = {60},
  number = {1},
  pages = {150--163},
  issn = {1558-0806},
  doi = {10.1109/TCSI.2012.2215780},
  eventtitle = {{{IEEE Transactions}} on {{Circuits}} and {{Systems I}}: {{Regular Papers}}},
  issue = {1}
}

@book{conseilinternationaldesgrandsreseauxelectriquesConnectionWindFarms2016,
  title = {Connection of Wind Farms to Weak {{AC}} Networks},
  editor = {{Conseil international des grands réseaux électriques}},
  year = {2016},
  publisher = {CIGRÉ},
  location = {Paris},
  isbn = {978-2-85873-374-3},
  langid = {english}
}

@standard{IEEEStandardInterconnection,
  title = {{{IEEE Standard}} for {{Interconnection}} and {{Interoperability}} of {{Inverter-Based Resources}} ({{IBRs}}) {{Interconnecting}} with {{Associated Transmission Electric Power Systems}}},
  institution = {IEEE},
  year = {2022},
  publisher = {IEEE},
  doi = {10.1109/IEEESTD.2022.9762253},
  isbn = {9781504484626},
  langid = {english}
}

@article{morrisAnalysisControllerBandwidth2021,
  title = {Analysis of {{Controller Bandwidth Interactions}} for {{Vector-Controlled VSC Connected}} to {{Very Weak AC Grids}}},
  author = {Morris, Jennifer F. and Ahmed, Khaled H. and {Egea-Alvarez}, Agusti},
  year = {2021},
  month = dec,
  journal = {IEEE Journal of Emerging and Selected Topics in Power Electronics},
  volume = {9},
  number = {6},
  pages = {7343--7354},
  issn = {2168-6777, 2168-6785},
  doi = {10.1109/JESTPE.2020.3031203},
  urldate = {2024-01-04},
  langid = {english}
}

@article{wuAssessingImpactRenewable2018a,
  title = {Assessing {{Impact}} of {{Renewable Energy Integration}} on {{System Strength Using Site-Dependent Short Circuit Ratio}}},
  author = {Wu, Di and Li, Gangan and Javadi, Milad and Malyscheff, Alexander M. and Hong, Mingguo and Jiang, John Ning},
  year = 2018,
  month = jul,
  journal = {IEEE Transactions on Sustainable Energy},
  volume = {9},
  number = {3},
  pages = {1072--1080},
  issn = {1949-3029, 1949-3037},
  doi = {10.1109/TSTE.2017.2764871},
  urldate = {2023-01-06},
  langid = {english}
}

@article{hatziargyriouDefinitionClassificationPower2021,
  title = {Definition and {{Classification}} of {{Power System Stability}} – {{Revisited}} \& {{Extended}}},
  author = {Hatziargyriou, Nikos and Milanovic, Jovica and Rahmann, Claudia and Ajjarapu, Venkataramana and Canizares, Claudio and Erlich, Istvan and Hill, David and Hiskens, Ian and Kamwa, Innocent and Pal, Bikash and Pourbeik, Pouyan and Sanchez-Gasca, Juan and Stankovic, Aleksandar and Van Cutsem, Thierry and Vittal, Vijay and Vournas, Costas},
  year = {2021},
  journal = {IEEE Transactions on Power Systems},
  shortjournal = {IEEE Trans. Power Syst.},
  volume = {36},
  number = {4},
  pages = {3271--3281},
  issn = {0885-8950, 1558-0679},
  doi = {10.1109/TPWRS.2020.3041774},
  langid = {english}
}

@article{xiaoDesignOrientedSmallSignalStability2025,
  title = {Design-{{Oriented Small-Signal Stability Assessment}} of {{Grid-Following Inverter-Based Resources Based}} on {{Virtual Dynamic Short-Circuit Ratio}}},
  author = {Xiao, Yu and Zhang, Xing and Zhan, Xiangdui and Cao, Renxian},
  year = 2025,
  journal = {IEEE Transactions on Industrial Electronics},
  pages = {1--12},
  issn = {1557-9948},
  doi = {10.1109/TIE.2025.3528509},
  urldate = {2025-02-17}
}

@article{zhuImpedanceMarginRatio2024,
  title = {Impedance {{Margin Ratio}}: A {{New Metric}} for {{Small-Signal System Strength}}},
  shorttitle = {Impedance {{Margin Ratio}}},
  author = {Zhu, Yue and Green, Timothy C. and Zhou, Xiaoyao and Li, Yitong and Kong, Dechao and Gu, Yunjie},
  year = 2024,
  journal = {IEEE Transactions on Power Systems},
  pages = {1--13},
  issn = {0885-8950, 1558-0679},
  doi = {10.1109/TPWRS.2024.3371231},
  urldate = {2024-04-18},
  copyright = {https://ieeexplore.ieee.org/Xplorehelp/downloads/license-information/IEEE.html},
  langid = {english}
}

@article{wuImpactNonMinimumPhaseZeros2021b,
  title = {Impact of {{Non-Minimum-Phase Zeros}} on the {{Weak-Grid-Tied VSC}}},
  author = {Wu, Guanglu and Zhao, Bing and Zhang, Xi and Wang, Shanshan and {Egea-Alvarez}, Agusti and Sun, Yuanyuan and Li, Yingbiao and Guo, Deyang and Zhou, Xiaoxin},
  year = 2021,
  month = apr,
  journal = {IEEE Transactions on Sustainable Energy},
  volume = {12},
  number = {2},
  pages = {1115--1126},
  issn = {1949-3029, 1949-3037},
  doi = {10.1109/TSTE.2020.3034791},
  urldate = {2025-01-09},
  copyright = {https://ieeexplore.ieee.org/Xplorehelp/downloads/license-information/IEEE.html},
  langid = {english}
}

@article{camposNovelGramianbasedStructurepreserving2022,
  title = {Novel {{Gramian-based Structure-preserving Model Order Reduction}} for {{Power Systems}} with {{High Penetration}} of {{Power Converters}}},
  author = {Campos, Nathalia de Morais Dias and Sarnet, Tanel and Kilter, Jako},
  year = 2022,
  journal = {IEEE Transactions on Power Systems},
  pages = {1--11},
  issn = {1558-0679},
  doi = {10.1109/TPWRS.2022.3228458},
  langid = {american}
}

@article{freudenberg1985right,
  author  = {Freudenberg, J. S. and Looze, D. P.},
  title   = {Right-half Plane Zeros and Poles and Design Tradeoffs in Feedback Systems},
  journal = {IEEE Transactions on Automatic Control},
  volume  = {AC-30},
  number  = {6},
  pages   = {555--565},
  month   = jun,
  year    = {1985}
}

@book{stewart1990matrix,
  author    = {Stewart, G. W. and Sun, Ji-guang},
  title     = {Matrix Perturbation Theory},
  series    = {Computer Science and Scientific Computing},
  publisher = {Academic Press},
  address   = {Boston, MA},
  year      = {1990},
  isbn      = {0-12-670230-6}
}

@article{Liu2024gOSCR,
  author={Liu, Chenxi and Xin, Huanhai and Wu, Di and Gao, Huisheng and Yuan, Hui and Zhou, Yuhan},
  journal={IEEE Transactions on Power Systems}, 
  title={Generalized Operational Short-Circuit Ratio for Grid Strength Assessment in Power Systems With High Renewable Penetration}, 
  year={2024},
  volume={39},
  number={4},
  pages={5479-5494},
  doi={10.1109/TPWRS.2023.3340158}}

@article{mohammedComparisonPLLBased2022,
  author  = {Nabil Mohammed and Weihua Zhou and Behrooz Bahrani},
  title   = {Comparison of {PLL}-Based and {PLL}-Less Control Strategies
             for Grid-Following Inverters Considering Time and Frequency
             Domain Analysis},
  journal = {IEEE Access},
  year    = {2022},
  volume  = {10},
  pages   = {80518--80538},
  doi     = {10.1109/ACCESS.2022.3195494}
}

@article{mohammedEnhancedFrequency2023,
  author  = {Nabil Mohammed and Mohammad Hasan Ravanji and Weihua Zhou
             and Behrooz Bahrani},
  title   = {Enhanced Frequency Control for Power-Synchronized
             {PLL}-Less Grid-Following Inverters},
  journal = {IEEE Open Journal of the Industrial Electronics Society},
  year    = {2023},
  volume  = {4},
  pages   = {189--204},
  doi     = {10.1109/OJIES.2023.3285010}
}

@article{liuOscillatoryStabilityCriterion2018,
  author  = {Liu, Huakun and Xie, Xiaorong and Liu, Wei},
  title   = {An Oscillatory Stability Criterion Based on the Unified
             {$dq$}-Frame Impedance Network Model for Power Systems With
             High-Penetration Renewables},
  journal = {IEEE Transactions on Power Systems},
  year    = {2018},
  volume  = {33},
  number  = {3},
  pages   = {3472--3485},
  month   = may,
  doi     = {10.1109/TPWRS.2018.2794067}
}
\vspace{-4.6em}
\begin{IEEEbiographynophoto}{Fuyilong Ma}
(Member, IEEE) received the B.Eng. and Ph.D. degrees from Zhejiang University, Hangzhou, China, in 2019 and 2024, respectively. Currently, he is a postdoc with Electric Power Research
Institute, China Southern Power Grid, Guangzhou, China. 
His research interests include power-electronic-dominated power system stability, dynamic networks and control theory. 
\end{IEEEbiographynophoto}
\vspace{-2em}
\begin{IEEEbiographynophoto}{Lidong Zhang}
(Senior Member, IEEE) received the B.Sc. from NCEPU, Baoding, China, in 1991, and Licentiate degree from Chalmers in 1999, and Ph.D. degrees from in electrical engineering from the Royal Institute of Technology (KTH), Stockholm, Sweden, in 2010, all in electrical engineering. From 1991 to 1997, he was with NCEPU, Beijing, China, from 2001 as a lecture of electrical engineering. From 1999 to 2010, he joined ABB Power Systems, Ludvika, Sweden, as a senior engineer. In 2010, he joined ABB Corporate Research, Västerås, where he was appointed as a Principal Scientist in 2012. In 2021, joined Hitachi Energy in Beijing China. In 2023, he joined China Southern Grid as a strategy specialist. His research interests include control and dynamic analysis of HVDC and FACTS, particularly grid-connected converters.
\end{IEEEbiographynophoto}
\vspace{-2em}
\begin{IEEEbiographynophoto}{Wangqianyun Tang}
 received the B.Eng. and Ph.D. degrees from the School of Electrical and Electronic Engineering, Huazhong University of Science and Technology, Wuhan, China, in 2014 and 2020, respectively. From 2020 to 2022, she was a Postdoctor with China Southern Power Grid Electric Power Research Institute (SEPRI). She is currently a Senior Engineer and Researcher with SEPRI. She is also an expert in IEC SC8A JWG5. Her research interests include modeling and control of renewable energy and stability analysis of power-electronized power system.
\end{IEEEbiographynophoto}
\vspace{-2em}
\begin{IEEEbiographynophoto}{Waisheng Zheng}
is currently a Professor level Senior Engineer and the Project Leader of National Natural Science Foundation of China. He is also the Chairman of SEPRI and the General Manager of strategic planning department, China Southern Power Grid, Guangzhou, China. His research interests include power system planning, scheduling, and smart grid.
\end{IEEEbiographynophoto}
\vspace{-2em}
\begin{IEEEbiographynophoto}{Huanhai Xin}
 (Senior Member, IEEE) received the Ph.D. degree from the College of Electrical Engineering, Zhejiang University, Hangzhou, China, in 2007. He was a Postdoctoral Researcher with the Department of Electrical Engineering and Computer Science, University of Central Florida, Orlando, FL, USA, from 2009 to 2010. He is currently a Professor with the College of Electrical Engineering, Zhejiang University. His research interests include renewable power system stability analysis and control. He was the recipient of the National Natural Science of China for Excellent Young Scholars.
\end{IEEEbiographynophoto}
\vspace{-2em}
\begin{IEEEbiographynophoto}{Linbin Huang}
 (Member, IEEE) received the B.Eng. and Ph.D. degrees from Zhejiang University, Hangzhou, China, in 2015 and 2020, respectively. From 2020 to 2024, he was a Senior Scientist with Automatic Control Laboratory, ETH Zurich, Zurich, Switzerland. He is currently a Professor with the College of Electrical Engineering, Zhejiang University. His research interests include power system stability, optimal control of power electron ics, and data-driven control.
\end{IEEEbiographynophoto}
\vspace{-2em}
\begin{IEEEbiographynophoto}{Lennart Harnefors}
 (Fellow, IEEE) received the M.Sc., Licentiate, and Ph.D. degrees in electrical engineering from the Royal Institute of Technology (KTH), Stockholm, Sweden, in 1993, 1995, and 1997, respectively, and the Docent (D.Sc.) degree in industrial automation from Lund University, Lund, Sweden, in 2000. From 1994 to 2005, he was with Mälardalen University, Västerås, Sweden, from 2001 as a Professor of Electrical Engineering. From 2001 to 2005, he was, in addition, a part-time Visiting Professor of Electrical Drives with the Chalmers University of Technology, Göteborg, Sweden. In 2005, he joined ABB, HVDC Product Group, Ludvika, Sweden, where, among other duties, he led the control development of the first generation of multilevel-converter HVDC Light. In 2012, he joined ABB, Corporate Research, Västerås, where he was appointed as a Senior Principal Scientist in 2013 and as a Corporate Research Fellow in 2021. His research interests include control and dynamic analysis of power electronic systems, particularly grid-connected converters, and ac drives. He is, in addition, a part-time Visiting Professor with Aalto University, Espoo, Finland. He was the recipient of the 2020 IEEE Modeling and Control Technical Achievement Award. He is an Editor of the IEEE Journal of Emerging and Selected Topics in Power Electronics.
\end{IEEEbiographynophoto}

\vfill
\end{document}